\documentclass[bj]{imsart}

\RequirePackage{amsthm,amsmath,amsfonts,amssymb}
\RequirePackage[numbers]{natbib}

\usepackage{times}
\usepackage{bm}
\usepackage[plain,noend]{algorithm2e}
\usepackage{url}

\usepackage{graphicx}
\usepackage{amsbsy}
\usepackage{mathrsfs}
\usepackage{lscape}
\usepackage{float}
\usepackage{multirow}
\usepackage{tablefootnote}
\usepackage{threeparttable}
\usepackage{pdflscape}
\usepackage{enumerate}
\usepackage{pdfpages}
\usepackage{soul}

\usepackage{xr}
\newcommand{\bbeta}{\boldsymbol{\beta}}
\newcommand{\bld}{{\boldsymbol{\Lambda}}}
\newcommand{\bX}{\mathbf{X}}
\newcommand{\bz}{\mathbf{z}}
\newcommand{\bA}{\mathbf{A}}

\newcommand{\baf}{\boldsymbol{\alpha}}
\newcommand{\bgam}{\boldsymbol{\gamma}}
\renewcommand{\P}{\mathrm{P}}
\newcommand{\dd}{\mathrm{d}}

\DeclareMathOperator{\sign}{sign}

\newcommand{\cond}[1]{($\mathcal C$.\ref{#1})}
\startlocaldefs
\theoremstyle{plain} 
\newtheorem{theorem}{Theorem}[section]
\newtheorem{lemma}[theorem]{Lemma}
\newtheorem{proposition}[theorem]{Proposition}

\endlocaldefs

\begin{document}
	
	\begin{frontmatter}
		\title{Adaptive lasso and Dantzig selector for spatial point processes intensity estimation}
		\runtitle{Lasso and Dantzig selector for spatial point processes}
		
		\begin{aug}
			\author[A]{\fnms{Achmad} \snm {Choiruddin}\ead[label=e1]{choiruddin@its.ac.id}},
			\author[B]{\fnms{Jean-Fran\c cois} \snm{Coeurjolly}\ead[label=e2,mark]{jean-francois.coeurjolly@univ-grenoble-alpes.fr}}
			\and
			\author[B]{\fnms{Frédérique} \snm{Letué}\ead[label=e3]{frederique.letue@univ-grenoble-alpes.fr}}
			\address[A]{Department of Statistics, Institut Teknologi Sepuluh Nopember (ITS), Indonesia, \printead{e1}}
			
			\address[B]{Laboratoire Jean Kuntzmann, Université Grenoble Alpes, France, \printead{e2} \printead{e3}}
		\end{aug}
		
		\begin{abstract}
			Lasso and Dantzig selector are standard procedures able to perform variable selection and estimation simultaneously. This paper is concerned with extending these procedures to spatial point process intensity estimation. We propose adaptive versions of these procedures, develop efficient computational methodologies and derive asymptotic results for a large class of spatial point processes under an original setting where the number of parameters, i.e. the number of spatial covariates considered, increases with the expected number of data points. Both procedures are compared theoretically, in a simulation study, and in a real data example.
		\end{abstract}
		
		
	\end{frontmatter}
	
	\section{Introduction} \label{sec:intro}
	
	Spatial point processes are stochastic processes which model random locations of points in space, such as random locations of trees in a forest, locations of disease cases, earthquake occurrences and crime events \citep[e.g.][]{baddeley2015spatial,choiruddin2021quantifying,coeurjolly2019understanding,moller2003statistical}. To understand the arrangement of points, the intensity function is the standard  summary function \cite{baddeley2015spatial,coeurjolly2019understanding}. When one seeks to describe the probability of observing a point at location $u\in \mathbb R^d$ in terms of covariates, the most popular model for the intensity function, $\rho$, is
	\begin{align}
		\label{eq:int}
		\rho (u;\bbeta)=\exp\{\bbeta^\top \bz(u)\}, \quad u \in D \subset \mathbb{R}^d,
	\end{align}
	where, for $p\ge 1$, $\bz(u)=\{ z_1(u),\ldots,z_p(u)\}^\top$, represent spatial covariates measured at location $u$ and $\bbeta=\{\beta_1,\ldots,\beta_p\}^\top$ is a real $p$-dimensional parameter.
	
	The score of the Poisson likelihood, i.e. the likelihood if the underlying process is assumed to be the Poisson process, remains an unbiased estimating equation to estimate $\bbeta$  even if the point pattern does not arise from a Poisson process \citep{waagepetersen2007estimating}. Such a method is well-studied in the literature and extended in several ways to gain efficiency \citep[e.g.][]{guan2015quasi,guan2010weighted} when the number of covariates is moderate. Standard results cover the consistency and asymptotic normality of the maximum Poisson likelihood under the increasing domain asymptotic (see e.g.~\cite{guan2015quasi} and references therein).
	
	When a large number of covariates is available, variable selection is unavoidable. Performing estimation and selection for spatial point processes intensity models has received a lot of attention. Recent developments consider techniques centered around regularization methods \citep[e.g.][]{choiruddin2018convex,daniel2018penalized,rakshit2021variable,yue2015variable} such as lasso technique. In particular, \cite{choiruddin2018convex} consider several composite likelihoods penalized by a large class of convex and non-convex penalty functions and obtain asymptotic results under the increasing domain asymptotic framework.
	
	The Dantzig selector is an alternative procedure to regularization techniques. It was initially proposed for linear models by~\cite{candes2007dantzig} and subsequently extended to more complex models \cite[e.g.][]{antoniadis2010dantzig,dicker2010regularized,james2009generalized}. In particular, \cite{dicker2010regularized} generalizes this approach to general estimating equations.
	One of the main advantages of the Dantzig selector is its implementation which, for linear models, results in a linear programming. Since then, the Dantzig selector and lasso procedures have been compared in different contexts \citep[e.g.][]{bickel2009simultaneous,james2009dasso}.
	
	In this paper, we compare lasso and Dantzig selector when they are applied to  intensity estimation for spatial point processes. We compare these procedures in the complex asymptotic framework where the number of informative covariates, say $s_n$ and the number of non-informative covariates, say $p_n-s_n$, may increase with the mean number of points. Our asymptotic results are developed under a setting which embraces both increasing domain asymptotic and infill asymptotic which are often considered in the literature (see also Section~\ref{sec:background}). Such a setting is almost never considered in the spatial point processes literature (see again e.g.~\cite{guan2015quasi}). 
	
	It is well-known that the Poisson likelihood can be approximated as a quasi-Poisson regression model, see e.g.~\cite{baddeley2015spatial}. For the adaptive lasso procedure, our theoretical contributions can also be seen as extension of work such as \cite{fan2004nonconcave} which provides asymptotic results for estimators from regularized generalized linear models. However, in our spatial framework, the more standard sample size must be substituted by, here, a mean number of points. Furthermore, observations are no more independent, i.e. our results are valid for a large class of spatial dependent point processes. Note also that \cite{fan2004nonconcave} assumes $s_n=s$.
	
	The theoretical contributions of the present paper are two-fold. First, for the adaptive lasso procedure, the question stems from extending the work by~\cite{choiruddin2018convex} which considers only an increasing domain asymptotic, and assumes $s_n=s$ and $p_n=p$. For the Dantzig selector, our contributions are to extend the standard methodology to spatial point processes, propose an adaptive version and derive theoretical results.
	This yields different computational and theoretical issues. As revealed by our main result, Theorem~\ref{thm:main}, the adaptive lasso and Dantzig selector procedures share several similarities but also some slight differences. 
	We first prove that both procedures satisfy an oracle property, i.e. methods can correctly select the nonzero coefficients with probability converging
	to one, and second that the estimators of the nonzero coefficients are asymptotically normal. However, the conditions under which results are valid for the Dantzig selector are slightly more restrictive.
	
	Our conducted simulation study and application to environmental data also demonstrate that both procedures behave similarly.

	\section{Background and framework} \label{sec:background}
	
	Let $\bX$ be a spatial point process on $\mathbb{R}^d$, $d\ge 1$. We view $\bX$ as a locally finite random subset of $\mathbb{R}^d$. Let $D \subset \mathbb{R}^d$ be a compact set of Lebesgue measure $|D|$ which will play the role of the observation domain. A realization of $\bX$ in $D$ is thus a set $\mathbf{x}=\{x_1, \ldots, x_m\}$, where $x \in D$ and $m$ is the observed number of points in $D$.  Suppose $\bX$ has intensity function $\rho$ and second-order product density $\rho^{(2)}$. Campbell theorem states that, for any function $k: \mathbb{R}^d \to [0,\infty)$ or $k: \mathbb{R}^d \times \mathbb{R}^d \to [0,\infty)$
	\begin{align}
		\label{eq:campbell}
		\mathbb{E}  \sum_{u \in \mathbf{X}} k(u)  ={\int_{\mathbb{R}^d} k(u) \rho (u)\mathrm{d}u}, \quad
		\mathbb{E} \sum_{u,v \in \mathbf{X}}^{\neq} k(u,v)=\int_{\mathbb{R}^d \times \mathbb{R}^d} k(u,v) \rho^{(2)} (u,v)\mathrm{d}u \mathrm{d}v.
	\end{align}
	Based on the first two intensity functions, the pair correlation function $g$ is defined by
	\begin{align*}
		g(u,v)=\frac{\rho^{(2)}(u,v)}{\rho(u)\rho(v)}, \quad u,v \in D,
	\end{align*}
	when both $\rho$ and $\rho^{(2)}$ exist with the convention $0/0=0$. The pair correlation function is a measure of departure of the model from the Poisson point process for which $g=1$. For further background materials on spatial point processes, see for example \cite{baddeley2015spatial,moller2003statistical}.
	
	For our asymptotic considerations, we assume that a sequence $(\bX)_{n\ge 1}$ is observed within a sequence of bounded domains $(D_n)_{n\ge 1}$. We denote by $\rho_n$ and $g_n$ the intensity and pair correlation of $\bX_n$. With an abuse of notation, we denote by $\mathbb E$ and $\mathrm{Var}$, the expectation and variance under $\bX_n$. we assume that the intensity $\rho_n$ writes $\rho_n(u) = \exp \{ \bbeta_0^\top \bz(u)\}$, $u\in D_n$. We thus let $\bbeta_0$ denote the true parameter vector and assume it can be decomposed as $\bbeta_0 = \{\beta_{01},\ldots,\beta_{0s_n},\beta_{0(s_n+1)},\ldots,\beta_{0p_n}\}^\top=(\bbeta^{\top}_{01},\bbeta^{\top}_{02})^\top = (\bbeta_{01}^\top, \mathbf 0^\top)^\top$. Therefore, $\bbeta_{01} \in \mathbb R^{s_n}$, $\bbeta_{02}=\mathbf 0 \in \mathbb R^{p_n-s_n}$ and $\bbeta_0 \in \mathbb R^{p_n}$, where $s_n$ is the number of non-zero coefficients, $p_n-s_n$ the number of zero coefficients and $p_n$ the total number of parameters. We underline that it is unknown to us which coefficients are non-zero and
	which are zero. Thus, we consider a sparse intensity model where in particular $s_n$ and $p_n$ may diverge to infinity as $n$ grows.
	
	For any $\bbeta\in \mathbb R^{p_n}$ or for the spatial covariates $\bz(u)$, we use a similar notation, i.e. $\bbeta=(\bbeta_1^\top,\bbeta_2^\top)^\top$ and $\bz(u)=\{\bz_1(u)^\top,\bz_2(u)^\top\}^\top$, $u\in D_n$. We let $\mu_n = \mathbb E\{N(D_n)\}$, that is the expected number of points in $D_n$. By Campbell theorem, we have
	\[
	\mu_n = \int_{D_n}  \rho_n(u,\bbeta_0) \dd u = \int_{D_n} \exp \{ \bbeta_0^\top \bz(u)\} \dd u
	=\int_{D_n} \exp \{ \bbeta_{01}^\top \bz_1(u)\}  \dd u
	\]
	Note that $\mu_n$ is a function of $D_n, \bbeta_{01}, \bz_1(u), s_n$. In this paper, we assume that $\mu_n\to \infty$ as $n\to \infty$. That kind of assumption is very general and embraces the well-known frameworks called increasing domain asymptotics and infill asymptotics. For the increasing domain context, $D_n \to \mathbb R^d$ and usually $\bbeta_{01}$ depends on $n$ only through $s_n$. For the infill asymptotics, $D_n=D$ is assumed to be a bounded domain of $\mathbb R^d$ and usually $z_1(u)=1$, $\bbeta_{01}=\theta_n \to \infty$ as $n\to \infty$. In some sense, the parameter $\mu_n$ plays the role of the sample size in standard inference.
	
	To reduce notation in the following, unless it is ambiguous, we do not index $\bX$, $\rho$, $g$, $\bbeta_0$, $\bbeta$, $\bz(u)$ with $n$.

	\section{Methodologies} \label{sec:method}

	\subsection{Standard methodology}

	If $\bX$ is a Poisson point process, then, on $D_n$, $\bX$ admits a density with respect to the unit rate Poisson point process \citep{moller2003statistical}. This yields the log-likelihood function for $\bbeta$, which, for the intensity model~\eqref{eq:int}, is proportional to
	\begin{equation}\label{eq:likepois}
		\ell_n(\bbeta)  
		= \sum_{u \in \mathbf{X} \cap D_n} \bbeta^\top \bz(u) - \int_{D_n} \rho(u; \bbeta)\mathrm{d}u. 
	\end{equation}
	The gradient of~\eqref{eq:likepois} writes 
	\begin{equation}
		\label{eq:Un}
		\mathbf{U}_n (\bbeta)=  \frac{\mathrm d}{\mathrm d \bbeta} \ell_n(\bbeta) = {\sum_{u \in \bX \cap D_n}\bz(u)} - {\int_{D_n} \bz(u) \rho(u; \bbeta)\mathrm{d}u}.
	\end{equation}
	If $\bX$ is not a Poisson point process, Campbell Theorem shows that \eqref{eq:Un} remains an unbiased estimating equation. Hence, the maximum of~\eqref{eq:likepois}, still makes sense for non-Poisson models. Such an estimator, which can be viewed as  composite likelihood has received a lot of attention in the literature and asymptotic properties are well-established when $p_n=p$ and $p$ is moderate \citep[e.g][]{guan2015quasi,guan2010weighted,waagepetersen2007estimating}.

	We end this section with the definition of the two following $p_n\times p_n$ matrices
	\begin{align}
		\mathbf{A}_n(\bbeta)&={\int_{D_n} \bz(u)\bz(u)^\top \rho(u;\bbeta)\mathrm{d}u} \label{eq:An} \\
		\mathbf{B}_n(\bbeta)&=\mathbf{A}_n(\bbeta) + {\int_{D_n} \int_{D_n} \bz(u)\bz(v)^\top \{g(u,v)-1\} \rho(u;\bbeta) \rho(v;\bbeta) \mathrm{d}u \mathrm{d}v}. \label{eq:Bn}
	\end{align}
	The matrix $\bA_n(\bbeta)$ corresponds to the sensitivity matrix defined by \linebreak$\bA_n(\bbeta)=-\mathbb E \{\mathrm d \mathbf U_n(\bbeta) / \mathrm d \bbeta^\top \}$ while $\mathbf{B}_n(\bbeta)$ corresponds to the variance of the estimating equation, i.e. $\mathbf{B}_n(\bbeta)= \mathrm{Var}\{ \mathbf U_n(\bbeta) \}$. By passing, we point out that $\bA_n(\bbeta) = -\mathrm d \mathbf U_n(\bbeta) / \mathrm d \bbeta^\top $.
	Let $\mathbf M_n$ be some $p_n\times p_n$ matrix, e.g. $\bA_n(\bbeta)$ or $\mathbf{B}_n(\bbeta)$. Such a matrix is decomposed as
	\begin{align}
		\label{partition}
		\mathbf{M}_n=
		\begin{bmatrix}
			\mathbf{M}_{n,1} \\
			\mathbf{M}_{n,2}
		\end{bmatrix}
		=
		\begin{bmatrix}
			\mathbf{M}_{n,11}  &  \mathbf{M}_{n,12}\\
			\mathbf{M}_{n,21} &  \mathbf{M}_{n,22}
		\end{bmatrix},
	\end{align}
	where $\mathbf{M}_{n,1}$ (resp. $\mathbf{M}_{n,2}$) is the first $ s_n \times p_n$ (resp. the following $ (p_n-s_n) \times p_n$) components of $\mathbf{M}_n$ and $ \mathbf{M}_{n,11}$ (resp. $ \mathbf{M}_{n,12}$, $ \mathbf{M}_{n,21}$, and $\mathbf{M}_{n,22}$) is the $s_n \times s_n$ top-left corner (resp. the $s_n \times (p_n-s_n)$ top-right corner, the $(p_n-s_n) \times s_n$ bottom-left corner, and the $(p_n-s_n) \times (p_n-s_n)$ bottom-right corner) of $\mathbf{M}_n$. In what follows, for a squared symmetric matrix $\mathbf{M}_n$, $\nu_{\min}(\mathbf M_n)$ and $\nu_{\max}(\mathbf M_n)$ denote respectively the smallest and largest eigenvalue of $\mathbf M_n$. Finally, $\|\mathbf y \|$ denotes the Euclidean norm of a vector $\mathbf y$, while $\|\mathbf M_n\| = \sup_{\|\mathbf y\|\neq 0} \|\mathbf M_n \mathbf y\|/\|\mathbf y\|$ denotes the spectral norm. We remind that the spectral norm is subordinate and that for a symmetric definite positive matrix $\|\mathbf M_n\|=\nu_{\max}(\mathbf M_n)$.

	\subsection{Adaptive lasso (AL)} \label{sec:al}
	
	When the number of parameters is large, regularization methods allow one to perform both estimation and variable selection simultaneously. When $p_n=p$, \cite{choiruddin2018convex} consider several regularization procedures which consist in adding a convex or non-convex penalty term to~\eqref{eq:likepois}. The proposed methods are unchanged even when the number of covariates diverges. In particular, the adaptive lasso consists in maximizing
	\begin{align} \label{regmed}
		Q_n( \bbeta)= \frac{1}{\mu_n}\ell_n( \bbeta) - {\sum_{j=1}^{p_n} \lambda_{n,j}|\beta_{j}|},
	\end{align}
	where the real numbers $\lambda_{n,j}$ are non-negative tuning parameters. We therefore define the adaptive lasso estimator as
	\begin{align}
		\hat \bbeta_{\mathrm{AL}}= \arg\max_{\bbeta \in \mathbb{R}^{p_n}} Q_n( \bbeta).
	\end{align}
	When $\lambda_{n,j}=0$ for $j=1,\dots,p_n$, the method reduces to the maximum composite likelihood estimator and when $\lambda_{n,j}=\lambda_n$ to the standard lasso estimator. If $\beta_1$ acts as an intercept, meaning that $z_1(u)=1, \forall u \in D_n$, it is often desired to let this parameter free. This can be done by setting $\lambda_{n,1}=0$ in the second term of \eqref{regmed}. Finally, the choice of $\mu_n$ as a normalization factor in~\eqref{regmed} 
	follows the implementation of the adaptive lasso procedure for generalized linear models in the standard software (e.g. \texttt{R} package \texttt{glmnet}~\cite{friedman2010regularization}).

	\subsection{Adaptive (linearized) Dantzig selector (ALDS)} \label{sec:alds}
	
	When applied to a likelihood \citep{candes2007dantzig,james2009generalized}, the Dantzig selector estimate is obtained by minimizing $\|\bbeta\|_1$ subject to the infinite norm of the score function {bounded} by some threshold parameter $\lambda$. In the spatial point process setting, we  propose an adaptive version of the Dantzig selector estimate as the solution of the problem
	{\begin{align}
			\label{MADSvec}
			\min \|\bld_n \bbeta \|_1 \mbox{ subject to } |(\mathbf U_n(\bbeta))_j| \le \lambda_{n,j} \quad \text{ for } j=1,\dots,p_n
		\end{align}
		where {$\bld_n=\mathrm{diag}(\lambda_{n,1},\cdots,\lambda_{n,p_n} )$} and where $\mathbf U_n(\bbeta)$ is the estimating equation given by~\eqref{eq:Un}. It is worth pointing out $\lambda_{n,j}=0$ for $j=1,\dots,p_n$ reduces  the criterion \eqref{MADSvec} to $\mathbf U_n(\bbeta)=0$, which leads to the maximum composite likelihood estimator. Similarly to the adaptive lasso procedure, the intercept can be let free by setting $\lambda_{n,1}=0$. However, in the following, we assume, for the ALDS procedure,  that $\lambda_{n,j}>0$ for convenience, in order to rewrite~\eqref{MADSvec} in the following matrix form
	}
	\begin{align}
		\label{MADS}
		\min \|\bld_n \bbeta \|_1 \mbox{ subject to } (\mu_n)^{-1} \Big \|\bld_n^{-1} \mathbf{U}_n(  {\bbeta})  \Big\|_\infty \leq  1.
	\end{align}
	{We claim that the whole methodology and the proofs could be redone without involving the notation $\bld_n^{-1}$ and so that Theorem~\ref{thm:main} remains valid if one does not regularize the intercept term for example.}
	
	Due to the nonlinearity of the constraint vector, standard linear programming can no more be used to solve \eqref{MADS}. This results in a non-convex optimization problem. In particular the feasible set $\{\bbeta : \| \bld_n^{-1} \mathbf U_n(\bbeta)\|_\infty \leq 1\}$ is non-convex, which makes the method difficult to implement and to analyze from a theoretical point of view. In the context of generalized linear models, \cite{james2009generalized} consider the iterative reweighted least squares method and define an iterative procedure where at each step of the algorithm the constraint vector corresponds to a linearization of the updated pseudo-score. Such a procedure is not straightforward to extend from~\eqref{MADS} and remains complex to analyze from a theoretical point of view. As an alternative, we  follow~\cite[][Chapter 3]{dicker2010regularized} and propose to linearize the constraint vector by expanding $\mathbf U_n(\bbeta)$ around $\tilde \bbeta$, an initial estimate of $\bbeta_0$, using a first order Taylor approximation; i.e. we substitute $\mathbf U_n(\bbeta)$ by $\mathbf U_n(\tilde \bbeta) + \mathbf A_n(\tilde \bbeta) (\tilde \bbeta - \bbeta)$. Such a linearization enables now the use of standard linear programming.
	We term 
	adaptive linearized Dantzig selector (ALDS) estimate and denote it by $\hat \bbeta_{\mathrm{ALDS}}$ the solution to the optimization problem
	\begin{align}
		\label{ADS2}
		\min \|\bld_n \bbeta \|_1 \mbox{ subject to }   (\mu_n)^{-1}\; \Big \|\bld_n^{-1} \big \{\mathbf{U}_n( \tilde {\bbeta}) +  \mathbf{A}_n( \tilde {\bbeta}) (\tilde {\bbeta}  -  {\bbeta}) \big \}  \Big\|_\infty \leq  1.
	\end{align}
	Properties of $\hat \bbeta_{\mathrm{ALDS}}$ depend on properties of $\tilde \bbeta$ which are made precise in the next section.
	
	\section{Asymptotic results} \label{sec:results}

	Our main result relies upon the following conditions:
	\begin{enumerate}[($\mathcal C$.1)]
		\item The intensity function has the log-linear specification given by~\eqref{eq:int} where $\bbeta \in \mathbb R^{p_n}$. \label{C:intensity}
		\item  $(\mu_n)_{n\ge 1}$ is an increasing sequence of real numbers, such that $\mu_n\to \infty$ as ${n} \to \infty$. \label{C:nun}
		\item The covariates $\mathbf z$ satisfy
		\[
		\sup_{n\geq 1} \; \sup_{i=1,\dots,p_n} \; \sup_{u \in \mathbb{R}^d} |z_i(u)| < \infty 
		\qquad \mbox{ and } \qquad
		\inf_{n\ge 1} \inf_{\boldsymbol \phi\in \mathbb R^{p_n}, \|\boldsymbol \phi\|=1} \inf_{u\in D_n} \{\boldsymbol \phi^\top \bz(u)\}^2 >0
		\] \label{C:cov}
		\item The intensity and pair correlation satisfy
		\[
		\int_{D_n}\int_{D_n} \rho(u;\bbeta_0)\rho(v;\bbeta_0)|g(u,v)-1| \dd u\dd v = O(\mu_n).
		\] \label{C:g}
		\item The matrix $B_{n,11}(\bbeta_0)$ satisfies 
		\[
		\liminf_{n} \inf_{\boldsymbol \phi \in \mathbb R^{s_n}, \|\boldsymbol \phi\|=1} \boldsymbol \phi^\top
		\big\{(\mu_n)^{-1}\mathbf{B}_{n,11}(\bbeta_0) \big\}\boldsymbol \phi>0.	
		\] \label{C:Bn}
		\item For any $\boldsymbol \phi \in \mathbb R^{s_n}\setminus \{0\}$, the following convergence holds in distribution as $n\to \infty$:
		\[
		\sigma_{\boldsymbol \phi}^{-1} \boldsymbol \phi^\top \mathbf U_{n,1}(\bbeta_0) \stackrel{d}{\to} N(0,1)
		\]
		where $\sigma^2_{\boldsymbol \phi} = \boldsymbol \phi^\top \mathbf{B}_{n,11}(\bbeta_0) \boldsymbol \phi$. \label{C:clt}
		\item The initial estimate $\tilde {\bbeta}$ satisfies $\|\tilde {\bbeta} - {\bbeta_0} \|=O_{\mathrm{P}}(\sqrt{p_n/{\mu_n}})$ and is such that $\|\mathbf A_{n,11}(\tilde \bbeta)^{-1}\|= O_\P(\mu_n^{-1})$. 
		\label{C:initial}
		\item As $n\to \infty$, we assume that $s_n,p_n$ and $\mu_n$ are such that as $n\to \infty$
		\[\left\{
		\begin{array}{ll}
			\max \left(\frac{p_n^4}{\mu_n} , \frac{s_n^2 p_n^3}{\mu_n}\right) \to 0& \quad \text{ for the AL estimate} \\
			\frac{s_n^3 p_n^4}{\mu_n} \to 0&\quad \text{ for the ALDS estimate.} 
		\end{array}\right.
		\] \label{C:snpn}
		\item Let  $a_n=\max_{j=1,\ldots,{s_n}} \lambda_{n,j}$ and $b_n=\min_{j={s_n}+1,\ldots,p_n} \lambda_{n,j}$. We assume that these sequences are such that, as $n \to \infty$
		\[\left\{
		\begin{array}{lll}
			a_n \sqrt{s_n \mu_n }\to 0, & b_n \sqrt{\frac{\mu_n}{p_n^2}} \to \infty &\quad \text{ for the AL estimate} \\
			a_n \sqrt{s_n^3 \mu_n }\to 0, & b_n \sqrt{\frac{\mu_n}{p_n^3}} \to \infty &\quad \text{ for the ALDS estimate.} 
		\end{array}\right.
		\]
		\label{C:anbn}
	\end{enumerate}

	Condition~\cond{C:intensity} specifies the form of intensity models considered in this paper. In particular, note that we do not assume that $\bbeta$ is an element of a bounded domain of $\mathbb R^{p_n}$. Condition~\cond{C:nun} specifies our asymptotic framework where we assume to observe in average more and more points in $D_n$. As already mentioned, this may cover increasing domain type or infill type asymptotics. To our knowledge, only \cite{choiruddin2021information} consider a similar asymptotic framework in order to construct information criteria for spatial point process intensity estimation. The context is however very different here as we consider a {large number of covariates} and we study methodologies (adaptive lasso or Dantzig) which are able to produce a sparse estimate.  
	Condition~\cond{C:cov} is quite standard and is not too restrictive. Note that conditions~\cond{C:intensity}-\cond{C:cov} allow us in Lemma~\ref{lem:rhoubeta} to prove that in a `neighbordhood' of $\bbeta_0$, $\int \rho(u;\bbeta)= O(\mu_n)$, a useful result widely used in our proofs. The last part of Condition~\cond{C:cov} asserts that at any location the covariates are linearly independent. Condition~\cond{C:cov} also implies first that $\liminf_{n} \inf_{\boldsymbol \phi, \|\boldsymbol \phi\|=1} \boldsymbol \phi^\top\big\{(\mu_n)^{-1}\mathbf{A}_{n,11}(\bbeta_0) \big\}\boldsymbol \phi>0$ and second that $\liminf_{n} \inf_{\boldsymbol \phi, \|\boldsymbol \phi\|=1} \boldsymbol \phi^\top\big\{(\mu_n)^{-1}\mathbf{A}_{n}(\bbeta_0) \big\}\boldsymbol \phi>0$. Condition~\cond{C:Bn} is a similar assumption but for the submatrix $\mathbf B_{n,11}(\bbeta_0)$ which corresponds to $\mathrm{Var}\{\mathbf U_{n,1}(\bbeta_0)\}$. Condition~\cond{C:g} is also natural. Combined with~Condition~\cond{C:cov}, this implies that $\mathrm\mathbf B_{n}(\bbeta_0)= O(\mu_n p_n)$. When $p_n=p$ (and therefore $s_n=s$) and in the increasing domain framework, such an assumption can be satisfied by a large class of spatial point processes such as determinantal point processes, log-Gaussian Cox processes {and Neyman-Scott point processes \cite[see][]{choiruddin2018convex}. When $p_n=p$ and in the infill asymptotic framework, these assumptions are also valid for many spatial point processes, as discussed by~\cite{choiruddin2021information}}. Condition~\cond{C:clt} is required to derive the asymptotic normality of $\hat \bbeta_{01}$. Under a specific framework, such a result has already been obtained for a large class of spatial point processes: by \cite{biscio:waagepetersen:19,waagepetersen2009two} under the increasing domain framework and $p_n=p$ and~\cite{choiruddin2017spatial} when $p_n\to \infty$; by~\cite{choiruddin2021information} in the infill/increasing domain asymptotics frameworks and $p_n=p$.
	
	Condition ($\mathcal C$.\ref{C:initial}) is very specific to the ALDS estimate which requires a preliminary estimate of $\bbeta$. That condition is not unrealistic  as a simple choice for $\tilde \bbeta$ could be the maximum of the composite likelihood function~\eqref{eq:likepois}, see the remark after Theorem~\ref{thm:main}. Of course, we do not require that $\tilde \bbeta$ produces a sparse estimate.
	
	Condition ($\mathcal C$.\ref{C:snpn}) reflects the restriction on the number of covariates that can be considered in this study. For the AL estimate, this assumption is very similar to the one required by~\cite{fan2004nonconcave} when $\mu_n$ is replaced by $n$ and where the number of non-zero coefficients $s_n$ is constant.
	
	Condition ($\mathcal C$.\ref{C:anbn}) contains the main ingredients to derive sparsity properties, consistency and asymptotic normality. We first note that if $\lambda_{n,j}=\lambda_n$, then $a_n=b_n=\lambda_n$, whereby it is easily deduced that the two conditions on $a_n$ and $b_n$ cannot be satisfied simultaneously even if $p_n=p$. This justifies the introduction of an adaptive version of the Dantzig selector and motivates the use of the adaptive lasso. The condition $a_n\sqrt{s_n \mu_n}\to 0$ for the adaptive lasso is similar to the one imposed by \cite{fan2004nonconcave} when $\mu_n$ is replaced by $n$ and $s_n=s$ in their context. However, we require a slightly stronger condition on $b_n$ than the one required by \cite{fan2004nonconcave}. In our setting, their assumption would be written as $b_n \sqrt{\mu_n/p_n} \to \infty$. However, we would have to assume that $\nu_{\max}\big(\mathbf{A}_{n}(\bbeta_0)\big)=O(\mu_n)$. Such a condition is not straightforwardly satisfied in our setting since, for instance, the conditions~\cond{C:nun}-\cond{C:g} 
	only imply that $\nu_{\max}\big(\mathbf{A}_{n}(\bbeta_0)\big)=O({p_n \mu_n})$.

	As already mentioned, we do not assume that $\tilde \bbeta$ satisfies any sparsity property. We believe this is the main reason why conditions~\cond{C:snpn})-\cond{C:anbn} contain slightly stronger assumptions for the ALDS estimate than for the AL estimate. We now present our main result, whose proof is provided in Appendices~\ref{sec:proofAL}-\ref{sec:proofALDS}.
	
	\medskip
	
	
	\begin{theorem}  
		\label{thm:main}
		Let $\hat \bbeta$ denote either $\hat \bbeta_{\mathrm{AL}}$ or $\hat \bbeta_{\mathrm{ALDS}}$. Assume that  the conditions ($\mathcal C$.1)-($\mathcal C$.\ref{C:anbn}) hold, then the following properties hold.
		\begin{enumerate}[(i)]
			\item $\hat \bbeta$ exists. Moreover $\hat \bbeta_{\mathrm{AL}}$ satisfies, 
			${\displaystyle
				\|\hat \bbeta_{\mathrm{AL}} -\bbeta_0\| = O_{\mathrm P} 
				\left( \sqrt{\frac{p_n}{\mu_n}} \right)}$.
			\item Sparsity: $\mathrm{P}(\hat \bbeta_{2}=0) \to 1$ as $n \to \infty$.
			\item Asymptotic Normality: for any $\boldsymbol \phi \in \mathbb R^{s_n}\setminus \{0\}$ such that $\|\boldsymbol \phi\|<\infty$
			\[
			\sigma_{\boldsymbol \phi}^{-1} \, \boldsymbol \phi^\top \mathbf A_{n,11}(\bbeta_0)
			(\hat \bbeta_{1}- \bbeta_{01})\xrightarrow{d} \mathcal{N}(0, 1)
			\]
		\end{enumerate}
		in distribution, where $\sigma^2_{\boldsymbol \phi} = \boldsymbol \phi^\top \mathbf{B}_{n,11}(\bbeta_0) \boldsymbol \phi$.
	\end{theorem}

	
	\medskip
	
	To derive the consistency of $\hat \bbeta_{\mathrm{AL}}$, a careful look at the proof in Appendix~\ref{sec:proofAL} shows that the condition $a_n \sqrt{{\mu_n s_n}/p_n} \to 0$ could be sufficient.  The convergence rate, i.e. $O_{\mathrm P} ( \sqrt{p_n/\mu_n})$, is $\sqrt{p_n}$ times the convergence rate of the estimator obtained when $p_n$ is constant \citep[see ][Theorem 1]{choiruddin2018convex}. It also corresponds to the rate of convergence obtained by \cite{fan2004nonconcave} for generalized linear models when $p_n\to \infty$ and when $\mu_n$ corresponds to the standard sample size. It is worth pointing out that a possible diverging number of non-zero coefficients $s_n$ does not affect the rate of convergence. It does however impose a more restrictive condition on $a_n$. 
	Still on Theorem~\ref{thm:main} (i), its proof shows that this result remains valid when $\lambda_{n,j}=0$ for  $j=1,\dots,p_n$. In other words, the maximum composite likelihood estimator is consistent with the same rate of convergence. Hence a simple choice for the initial estimate $\widetilde \bbeta$ defining the ALDS estimate is the maximum of the Poisson likelihood given by~\eqref{eq:likepois}.
	
	Theorem~\ref{thm:main} (iii) would be the result one would obtain if $p_n-s_n=0$. Therefore, the efficiency  of $\hat \bbeta_{\mathrm{AL},1}$ and $\hat \bbeta_{\mathrm{ALDS},1}$ is the same as the estimator of $\bbeta_{01}$ obtained by maximizing \eqref{eq:likepois} based on the sub model knowing that $\bbeta_{02}=\mathbf{0}$. In other words, when $n$ is sufficiently large, both estimators are as efficient as the oracle one. 
	
	We end this section with the following remark. Despite the asymptotic properties for $\hat \bbeta_{\mathrm{AL}}$ and $\hat \bbeta_{\mathrm{ALDS}}$  and the conditions under which they are valid,  are (almost) identical, the proofs are completely different and rely upon different tools. For $\hat \bbeta_{\mathrm{AL}}$, our contribution is to extend the proof by~\cite{choiruddin2018convex} where only the increasing domain framework was considered, i.e. $\mu_n=O(|D_n|)$ and $D_n\to \mathbb R^d$ as $n\to \infty$ and where $s_n=s$ and $p_n=p$. The results for $\hat \bbeta_{\mathrm{ALDS}}$ are the first ones available for spatial point processes. To handle this estimator, we first have to study existence and optimal solutions for the primal and dual problems.  
	
	\section{Computational considerations} \label{sec:comp}
	
	To lessen notation, we remove the index $n$ in different quantities such as $\mathbf U_n, \lambda_{n,j}, \ell_n$, and $\mu_n$. For the practical point of view, we define $\mu=N(D)$ the number of data points in $D$.
	
	\subsection{Berman-Turner approach}
	
	Before discussing how AL and ALDS estimates are obtained we first remind the Berman-Turner approximation \citep{baddeley2015spatial} used to derive the Poisson likelihood estimate~\eqref{eq:likepois}. The so-called Berman-Turner approximation consists in discretizing the integral term in \eqref{eq:likepois} as
	\begin{align*}
		{\int_{D}  \rho(u; \bbeta) \mathrm{d}u} \approx {\sum_{i=1}^{M} w(u_i) \rho (u_i; \bbeta)},
	\end{align*}
	where $u_i, i=1,\ldots,M$ are points in $D$ consisting of the $m$ data points and $M-m$ dummy points and where the quadrature weights $w(u_i)>0$ are non-negative real numbers such that ${\sum_i w(u_i)}=|D|$. Using this specific integral discretization, \eqref{eq:likepois} is then approximated by
	\begin{align}
		\label{eq:appx:pois}
		\ell(\bbeta) \approx \tilde \ell(\bbeta) = {\sum_{i=1}^{M} w_i \{y_i \log \rho_i(\bbeta) - \rho_i(\bbeta)\}},
	\end{align}
	where $w_i=w(u_i), y_i=w_i^{-1} \mathbf{1}(u_i \in \bX \cap D)$ and $\rho_i(\bbeta)=\rho (u_i; \bbeta)$.
	Equation \eqref{eq:appx:pois} is formally equivalent to the weighted likelihood function of independent Poisson variables $y_i$ with weights $w_i$. The method to approximate \eqref{eq:Un} and \eqref{eq:An} follows along similar lines which respectively results in
	\begin{align}
		\mathbf{U}(\bbeta) \approx \tilde{\mathbf{U}}(\bbeta) = {\sum_{i=1}^{M} w_i   \bz_i \{y_i - \rho_i(\bbeta)\}}, \quad
		\mathbf{A}(\bbeta) \approx \tilde{\mathbf{A}}(\bbeta) = {\sum_{i=1}^{M} w_i   \bz_i \bz_i^\top \rho_i(\bbeta)}, \label{eq:approx:An}
	\end{align}
	where $\bz_i=\bz(u_i)$. Thus, standard statistical software for generalized linear models can be used to obtain the estimates. This fact is implemented in the $\mathtt{spatstat}$ $\texttt{R}$ package by $\texttt{ppm}$ function with option $\texttt{method="mpl"}$ \citep{baddeley2015spatial}.
	
	\subsection{Adaptive lasso (AL)} \label{sec:comp:AL}
	
	First, given a current estimate $\check \bbeta$, \eqref{eq:appx:pois} is approximated using second order Taylor approximation to apply iteratively reweighted least squares,
	\begin{align}
		\tilde \ell(\bbeta) \approx \ell_Q(\boldsymbol {\beta})
		=- \frac{1}{2} {\sum_{i=1}^M \psi_i (y_i^*-\bbeta^\top \bz_i)^2+C(\check \bbeta)} \label{eq:quad},
	\end{align}
	where $C(\check \bbeta)$ is a constant, $y_i^*$ and $\psi_i$ are the new working response values and weights,
	$
	y_i^*=\bz_i^\top \check \bbeta+\{y_i - \exp(\check \bbeta^\top\bz_i)\}/\{\exp(\check \bbeta^\top\bz_i)\}, \;\psi_i= w_i \exp(\check \bbeta^\top\bz_i). 
	$
	Second, a penalized weighted least squares problem is obtained by adding penalty term. Therefore, we solve
	\begin{align}
		\label{eq:glmnet}	
		{\displaystyle \min_{\bbeta \in \mathbb{R}^{p}} \Omega(\bbeta)}={\displaystyle \min_{\bbeta \in \mathbb{R}^{p}} \left\{-\frac{1}{N(D)}\ell_Q(\bbeta)+\sum_{j=1}^{p} \lambda_j|\beta_j|\right\}}
	\end{align}
	using the coordinate descent algorithm \citep{friedman2010regularization}. The method consists in partially minimizing ($\ref{eq:glmnet}$) with respect to $\beta_j$ given $\check \beta_l$ for $l \neq j$, $l,j=1,\ldots,p$, that is
	\begin{align*}	
		{\displaystyle \min_{\beta_j} \Omega (\check \beta_1, \ldots, \check \beta_{j-1}, \beta_j, \check \beta_{j+1}, \ldots, \check \beta_{p})}.
	\end{align*}
	With a few modifications \eqref{eq:glmnet} is solved using the $\mathtt{glmnet}$ $\texttt{R}$ package \citep{friedman2010regularization}. More detail about this implementation can be found in~\cite[][Appendix C]{choiruddin2018convex}.
	
	\subsection{Adaptive (linearized) Dantzig selector (ALDS)} \label{sec:comp:ALDS}
	
	Given $\tilde {\bbeta}$, \eqref{ADS2} is a linear problem, simple to implement. The main task is to compute the vectors $\mathbf{U}(\tilde \bbeta)$ and $\mathbf{A}(\tilde \bbeta)(\tilde {\bbeta}- {\bbeta})$. Typically, $\tilde {\bbeta}$ is chosen from the maximum composite likelihood estimate. Then, $\mathbf{U}(\tilde \bbeta)$ and $\mathbf{A}(\tilde \bbeta)(\tilde {\bbeta}- {\bbeta})$ are approximated by \eqref{eq:approx:An}. This results in solving
	\begin{align*}
		\min {\sum_{j=1}^{p} \lambda_j |\beta_j|} \mbox{ subject to } {N(D)}^{-1}\Big |\tilde{\mathbf{U}}_{j}(\tilde \bbeta) + \big\{\tilde{\mathbf{A}}(\tilde \bbeta)(\tilde {\bbeta}- {\bbeta})\big\}_j \Big | \leq \lambda_j, \quad \text{ for } j=1,\dots,p_n,
	\end{align*}
	where $ \tilde{\mathbf{U}}_{j}(\tilde \bbeta)$ and $ \{\tilde{\mathbf{A}}(\tilde \bbeta)(\tilde {\bbeta}- {\bbeta})\}_j$ are the $j$-th components of vectors $\tilde{\mathbf{U}}(\tilde \bbeta)$ and $\tilde{\mathbf{A}}(\tilde \bbeta)(\tilde {\bbeta}- {\bbeta})$.
	
	\subsection{Tuning parameter selection}
	\label{tuning}
	
	Both AL and ALDS rely on the proper regularization parameters $\lambda_j$ to avoid from unnecessary bias due to too small $\lambda_j$ selection and from large variance due to too large $\lambda_j$ choice, so the selection of $\lambda_j$ becomes an important task. To tune the $\lambda_j$, we follow \cite{choiruddin2018convex,zou2006adaptive} and define $\lambda_j= \lambda |\tilde \beta_j|^{-\nu}$, where $\lambda\geq 0, \nu>0$ and $\tilde {\bbeta}$ is the maximum composite likelihood estimate. The weights $|\tilde \beta_j|^{-\nu}$ serve as a prior knowledge to identify the non-zero coefficients since, for a constant $\lambda$, large (resp. small) $\tilde \beta_j$ will force $\lambda_j$ close to zero (resp. infinity) \cite{zou2006adaptive},

	{Our theoretical results are not in line with this stochastic way of setting the regularization parameters but we believe this choice is pertinent in our context. Here is the intuition: let $\lambda_{nj}=\lambda_n/|\tilde \beta_j|$. Using the $\sqrt{\mu_n/p_n}$-consistency of $\tilde \beta_j$, $\lambda_{nj}=O_P(\lambda_n)$ for non-zero coefficients, while, for zero coefficients, we may conjecture that $\lambda_{nj} = O_P(\lambda_n \sqrt{\mu_n/p_n})$. Forgetting the $O_P$, we may conjecture that $a_n\asymp \lambda_n$ and $b_n\asymp\lambda_n \sqrt{\mu_n/p_n}$. Hence, considering the AL procedure for instance, that would mean that we require $\lambda_n$  to be such that $\lambda_n \sqrt{s_n\mu_n}\to 0$ and $\lambda_n \mu_n/p_n^{3/2} \to \infty$, which constitutes a non empty  condition. For instance, assuming~($\mathcal C$.\ref{C:snpn}), the sequence $\lambda_n=(s_n\mu_n)^{-\eta}$ satisfies~($\mathcal C$.\ref{C:anbn}) as soon as $1/2<\eta<5/9$, since as $n\to \infty$
		\[
		\lambda_n \sqrt{s_n\mu_n} = (s_n\mu_n)^{1/2-\eta}\to 0  
		\qquad \text{ and } \qquad 
		\frac{\lambda_n \mu_n}{p_n^{3/2}}=\left(\frac{\mu_n}{s_n^2 p_n^3}\right)^{\eta/2} \left(\frac{\mu_n}{p_n^4}\right)^{1-3\eta/2} 
		p_n^{5/2-9\eta/2}	 \to \infty.
		\] 
		Considering rigorously the stochastic choice $\lambda_n/|\tilde \beta_j|$ for the regularization parameters is left for future research.}
	
	The remaining task is to specify $\lambda$. We follow here the literature, mainly \cite{choiruddin2021information}, and propose to select $\lambda$ as the minimum of the Bayesian information criterion,  BIC($\lambda$), for spatial point processes
	\begin{align*}
		\mathrm{BIC}(\lambda)=-2 \ell\{\hat{\bbeta} (\lambda)\} +  p_* \log N(D),
	\end{align*}
	where $\ell\{\hat{\bbeta}  (\lambda)\}$ is the maximized composite likelihood function, $p_*$ is the number of non-zero elements in $\hat \bbeta(\lambda)$ and $N(D)$ is the number of observed data points. 
	
	\section{Numerical results} \label{sec:num}
	
	In Sections~\ref{sec:sim}-\ref{sec:appl}, we compare the AL and ALDS for intensity modeling of a simulated and real data. In particular, the real data example comes from an environmental data where 1146 locations of Acalypha diversifolia trees given by Figure~\ref{fig:acaldi} are surveyed in a 50-hectare region ($D =1000m \times 500m$) of the tropical forest of Barro Colorado Island (BCI) in central Panama \cite[e.g.][]{hubbell2005barro}. A main question is how this tree species profits from environmental habitats \cite{choiruddin2020regularized,waagepetersen2009two} that could be related to the 15 environmental covariates depicted in Figure~\ref{cov} and their 79 interactions. With a total of 94 covariates, we perform variable selection using the AL and ALDS to determine which covariates should be included in the model. We center and scale the 94 covariates to sort the important covariates according to the magnitudes of $\hat \bbeta$. The subset of such covariates are also considered to construct a realistic setting for the simulation studies.
	
	\begin{figure}[ht]
		\includegraphics[width=0.81\textwidth]{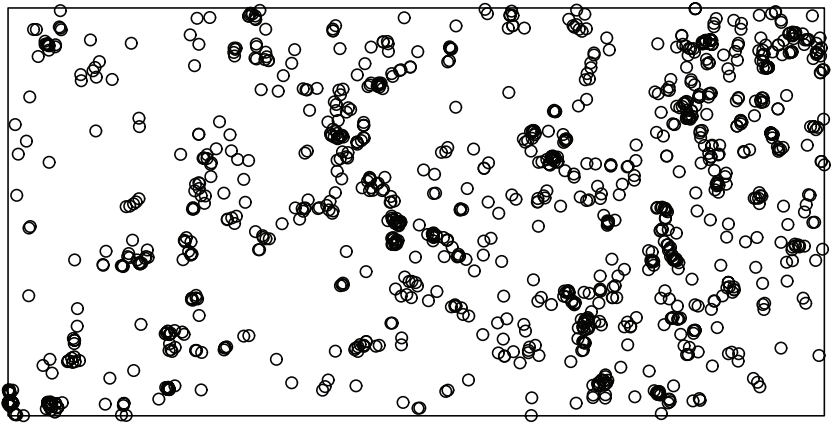}
		\caption{Plot of 1146 Acalypha diversifolia tree locations observed in the tropical forest of Barro Colorado Island}
		\label{fig:acaldi}
	\end{figure}
	
	\begin{figure}[ht]
		\renewcommand{\arraystretch}{0}
		\centering
		\begin{tabular}{l l l}
			\includegraphics[width=0.31\textwidth]{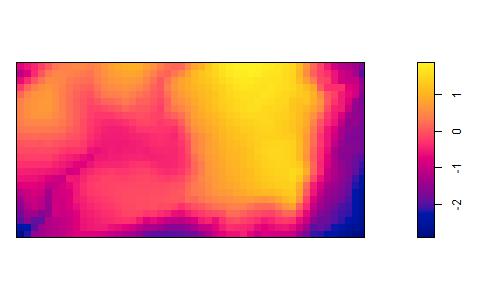} & \includegraphics[width=0.31\textwidth]{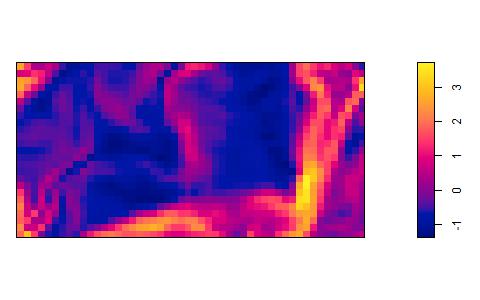} &  \includegraphics[width=0.31\textwidth]{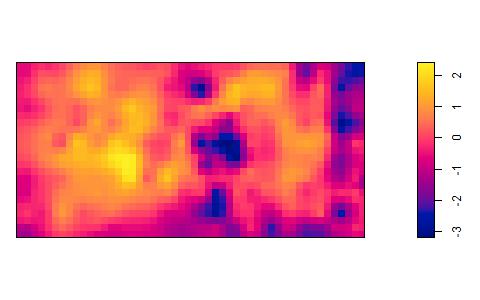} \\ \includegraphics[width=0.31\textwidth]{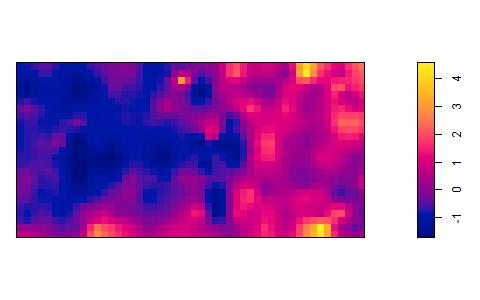} &  \includegraphics[width=0.31\textwidth]{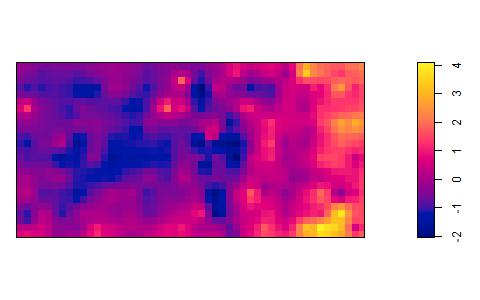} &
			\includegraphics[width=0.31\textwidth]{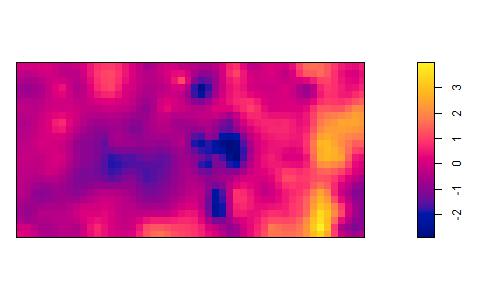} \\ \includegraphics[width=0.31\textwidth]{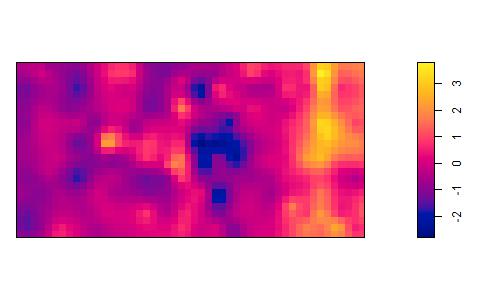} &  \includegraphics[width=0.31\textwidth]{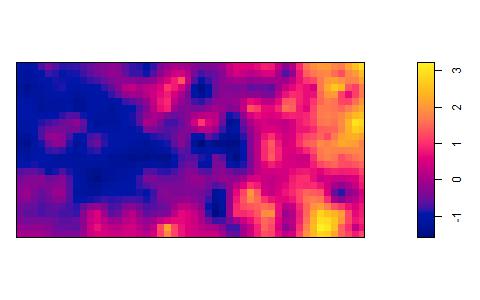} &  \includegraphics[width=0.31\textwidth]{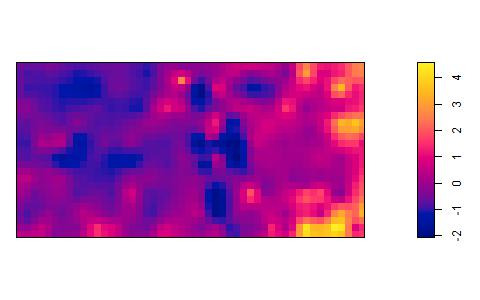} \\  \includegraphics[width=0.31\textwidth]{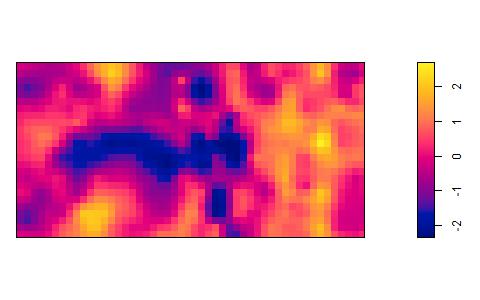} &
			\includegraphics[width=0.31\textwidth]{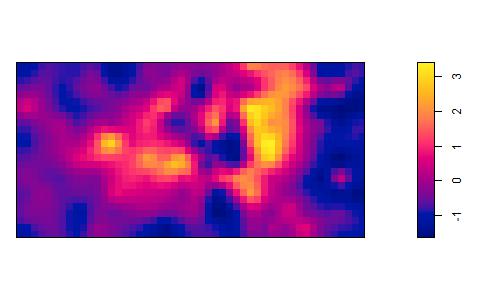} & \includegraphics[width=0.31\textwidth]{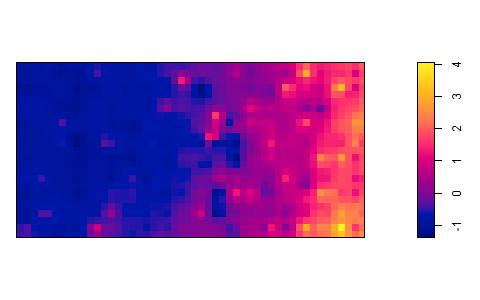} \\  \includegraphics[width=0.31\textwidth]{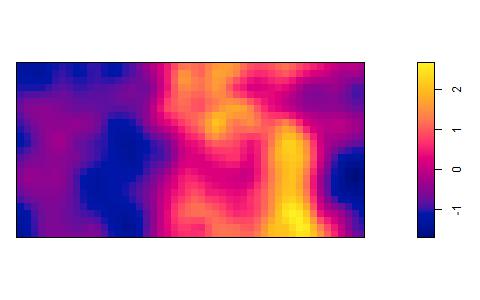} &  \includegraphics[width=0.31\textwidth]{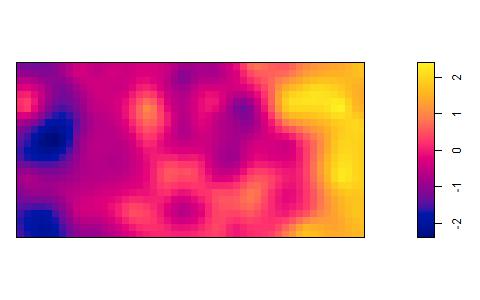} &  \includegraphics[width=0.31\textwidth]{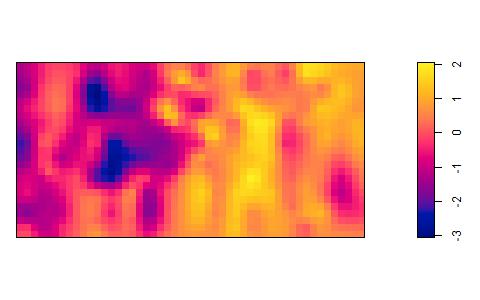}\\
		\end{tabular}
		\caption{Maps of covariates used in simulation study and in application. From left to right: Elevation, slope, Aluminium (row 1), Boron, Calcium, Copper (row 2), Iron, Potassium, Magnesium (row 3), Manganese, Phosporus, Zinc (row 4), and Nitrogen, Nitrigen mineralisation, pH (row 5).}
		\label{cov}
	\end{figure}
	
	\subsection{Simulation study} \label{sec:sim}
	
	The simulated point patterns are generated from Poisson and Thomas cluster processes with intensity~\eqref{eq:int}. To generate point patterns from Thomas process with intensity \eqref{eq:int} \cite[e.g.][]{choiruddin2018convex}, we first generate a parent point pattern from a stationary Poisson point process $\mathbf{C}$ with intensity $\kappa=4 \times10^{-4}$. Given $\mathbf{C}$, offspring point patterns are generated from inhomogeneous Poisson point process $\mathbf{X}_c, c \in \mathbf{C}$ with intensity
	\begin{align*}
		\rho_{child}(u)=\exp\{\bbeta^\top \mathbf{z}(u)\} k(u-c;\gamma)/\kappa,
	\end{align*}	
	where $k(u-c;\gamma)=(2 \pi \gamma^2)^{-1} \exp(-\|u-c\|^2/(2 \gamma^2))$. 
	The point process $\mathbf{X}=\cup_{c \in \mathbf{C}}\mathbf{X}_c$ is indeed an inhomogeneous Thomas point process with intensity \eqref{eq:int}. We set $\gamma=5$ and $15$. The smaller the $\gamma$, the more clustered the point patterns, leading to moderate clustering for $\gamma=15$ and highly clustering for $\gamma=5$.
	
	The covariates $\bz(u)$ used for the simulation experiment are from the BCI data. In addition to the 15 environmental factors depicted by Figure~\ref{cov}, we add interaction between two covariates until we obtain the desired number of covariates $p=21, 41$ or $81$. For each $p$, we consider three mean numbers of points ($\mu$) which increase when the observation domain expands. More precisely, we set that $\mu_1=150$ (resp. $\mu_2=600,\mu_3=2400$) points is generated in $D_1=[0,250]\times[0,125]$ (resp.  $D_2=[0,500]\times [0,250]$, $D_3=[0,1000]\times[0,500]$). When $D_1$ or $D_2$ is considered, we simply rescale the covariates to fit the desired observation domain. We fix $\beta_2=1$ and $\beta_3=-1$ while the rest are set to zero, {so $s=2$}. The parameter $\beta_1$ acts as intercept and is tuned to control the average number of points.
	
	\begin{table}[ht]
		\centering
		\setlength{\tabcolsep}{7pt}
		\renewcommand{\arraystretch}{1.1}
		\caption{True positive rate (TPR), false positive rate (FPR) in percentage, RMSE and average time in seconds obtained for AL and ALDS estimates based on 500 simulations from inhomogeneous Poisson point processes observed on different observation domains.}
		\label{tab:poisson}
		\begin{tabular}{rrrrrrrrr}
			\hline
			& \multicolumn{2}{c}{TPR} & \multicolumn{2}{c}{FPR} &\multicolumn{2}{c}{RMSE} & \multicolumn{2}{c}{Time}\\
			& AL & ALDS & AL & ALDS & AL & ALDS & AL & ALDS \\ 
			\hline
			\multicolumn{1}{l}{$D_1$ ($\mu_1=150$)} &&&&&&&&\\
			$p=20$ &  57 & 57 & 23 & 23 & 2.4 & 2.4 & 0.3 & 0.3 \\
			$p=40$ & 7 & 7 & 15 & 15 & 2.9 & 2.9 & 2.0 & 2.0\\
			$p=80$ & 0 & 0 & 8 & 8 & 2.8 & 2.8 & 4.0 & 4.0 \\
			\multicolumn{1}{l}{$D_2$ ($\mu_2=600$)} &&&&&&&&\\
			$p=20$ &  100 & 100 & 3 & 4 & 0.3 & 0.3 & 0.3 & 0.3\\
			$p=40$ & 97 & 96 & 4 & 5 & 0.5 & 0.5 & 0.8 & 0.6 \\
			$p=80$ & 86 & 86 & 8 & 8 & 0.9 & 0.9 & 3.0 & 3.0\\
			\multicolumn{1}{l}{$D_3$ ($\mu_3=2400$)} &&&&&&&&\\
			$p=20$ & 100 & 100 & 0 & 0 & 0.1 & 0.1 & 0.4 & 0.3\\
			$p=40$ & 100 & 100 & 0 & 0 & 0.1 & 0.1 & 0.9 & 0.9\\
			$p=80$ & 100 & 100 & 0 & 0 & 0.1 & 0.1 & 3.0 & 3.0 \\
			\hline
		\end{tabular}
	\end{table}
	
	\begin{table}[ht]
		\centering
		\setlength{\tabcolsep}{7pt}
		\renewcommand{\arraystretch}{1.1}
		\caption{True positive rate (TPR), false positive rate (FPR) in percentage, RMSE and average time in seconds obtained for AL and ALDS estimates based on 500 simulations from inhomogeneous Thomas point processes with $\kappa=4\times 10^{-4}$ and $\gamma=15$ (moderate clustering) observed on different observation domains.}
		\label{tab:MC}
		\begin{tabular}{rrrrrrrrr}
			\hline
			& \multicolumn{2}{c}{TPR} & \multicolumn{2}{c}{FPR} &\multicolumn{2}{c}{RMSE} & \multicolumn{2}{c}{Time}\\
			& AL & ALDS & AL & ALDS & AL & ALDS & AL & ALDS \\ 
			\hline
			\multicolumn{1}{l}{$D_1$ ($\mu_1=150$)} &&&&&&&&\\
			$p=20$ & 59 & 59 & 44 & 44 & 8.2 & 8.2 & 0.3 & 0.3 \\ 
			$p=40$ & 20 & 20 & 30 & 30 & 11.0 & 11.0 & 2.0 & 2.0 \\ 
			$p=80$ & 0 & 0 & 15 & 15 & 8.3 & 8.3 & 4.0 & 4.0 \\ 
			\multicolumn{1}{l}{$D_2$ ($\mu_2=600$)} &&&&&&&&\\
			$p=20$ & 91 & 89 & 53 & 47 & 2.6 & 2.2 & 0.3 & 0.3 \\ 
			$p=40$ & 88 & 86 & 48 & 43 & 4.9 & 3.8 & 1.0 & 0.6 \\ 
			$p=80$ & 80 & 80 & 35 & 35 & 7.7 & 7.7 & 5.0 & 5.0 \\ 
			\multicolumn{1}{l}{$D_3$ ($\mu_3=2400$)} &&&&&&&&\\
			$p=20$ & 100 & 100 & 59 & 43 & 1.0 & 0.8 & 0.9 & 1.0 \\   
			$p=40$ & 100 & 100 & 56 & 39 & 1.7 & 1.1 & 2.0 & 3.0 \\ 
			$p=80$ & 100 & 100 & 52 & 52 & 3.2 & 3.2 & 8.0 & 8.0 \\ 
			\hline
		\end{tabular}
	\end{table}
	
	\begin{table}[!ht]
		\centering
		\setlength{\tabcolsep}{7pt}
		\renewcommand{\arraystretch}{1.1}
		\caption{True positive rate (TPR), false positive rate (FPR) in percentage, RMSE and average time in seconds obtained for AL and ALDS estimates based on 500 simulations from inhomogeneous Thomas point processes with $\kappa=4\times 10^{-4}$ and $\gamma=5$ (high clustering) observed on different observation domains.}
		\label{tab:HC}
		\begin{tabular}{rrrrrrrrr}
			\hline
			& \multicolumn{2}{c}{TPR} & \multicolumn{2}{c}{FPR} &\multicolumn{2}{c}{RMSE} & \multicolumn{2}{c}{Time}\\
			& AL & ALDS & AL & ALDS & AL & ALDS & AL & ALDS \\ 
			\hline
			\multicolumn{1}{l}{$D_1$ ($\mu_1=150$)} &&&&&&&&\\
			$p=20$ & 64 & 64 & 72 & 72 & 39.0 & 39.0 & 0.5 & 0.5 \\ 
			$p=40$ & 29 & 29 & 72 & 72 & 130.0 & 130.0 & 6.0 & 6.0 \\ 
			$p=80$ &  0 & 0 & 36 & 36 & 69.0 & 69.0 & 6.0 & 6.0 \\ 
			\multicolumn{1}{l}{$D_2$ ($\mu_2=600$)} &&&&&&&&\\
			$p=20$ & 91 & 90 & 66 & 61 & 4.1 & 3.6 & 0.4 & 0.3 \\ 
			$p=40$ & 84 & 80 & 70 & 64 & 13.0 & 10.0 & 2.0 & 0.7 \\ 
			$p=80$ & 80 & 80 & 74 & 74 & 53.0 & 53.0 & 10.0 & 10.0 \\ 
			\multicolumn{1}{l}{$D_3$ ($\mu_3=2400$)} &&&&&&&&\\
			$p=20$ & 100 & 100 & 66 & 51 & 1.4 & 1.1 & 1.0 & 1.0 \\ 
			$p=40$ & 100 & 100 & 69 & 51 & 2.6 & 1.9 & 3.0 & 4.0 \\ 
			$p=80$ & 100 & 100 & 70 & 70 & 7.1 & 7.1 & 10.0 & 10.0 \\ 
			\hline
		\end{tabular}
	\end{table}

	For each model and setting, we generate 500 independent point patterns and estimate the parameters for each of these using the AL and ALDS procedures. The performances of AL and ALDS estimates are compared in terms of the true
	positive rate (TPR),  false positive rate (FPR), and root-mean squared error (RMSE). We also report the computing time. The TPR (resp.\ FPR) are the
	expected fractions of informative (resp.\ non-informative)
	covariates included in the selected model, so we would expect to obtain a high (resp. low) TPR (resp. FPR).
	The RMSE is defined for an estimate $\hat \bbeta$ by
	\begin{align*}
		\mathrm{RMSE}=\left\{ \sum_{j=2}^{p} { \hat{\mathbb{E}}(\hat \beta_j-\beta_j)^2} \right\}^\frac{1}{2}
	\end{align*}
	where $\hat{\mathbb{E}}$ is the empirical mean. 
	
	Tables~\ref{tab:poisson}-\ref{tab:HC} report results respectively for the Poisson model and the Thomas model with moderate and high clustering. 
	In the situation where point patterns come from Poisson processes, AL and ALDS perform very similarly. In particular, both methods do not work well in a small spatial domain with large $p$. The performances improve significantly as $D$ expands even for large $p$. When the point patterns exhibit clustering (Tables~\ref{tab:MC}-\ref{tab:HC}), in general AL and ALDS tend to overfit the intensity model by selecting too many covariates (indicated by higher FPR) which yields higher RMSE. ALDS sometimes performs slightly better in terms of RMSE. Results tend to deteriorate in the high clustering situation but  remain very satisfactory: for the three considered models, the TPR (resp. FPR, RMSE) increases (decreases) when $|D|$ grows for given $p$ while for given $D$ (especially when $D=D_2,D_3$), the results remain quite stable when $p$ increases. In terms of computing time, no major difference can be observed.
	
	\subsection{Application to the forestry dataset} \label{sec:appl}
	
	We model the intensity function for the point pattern of the Acalypha diversifolia using~\eqref{eq:int} depending on 94 environmental described previously. The overall $\bbeta$ estimates from the AL and ALDS procedures are presented in Table~\ref{tab:betaest}. We only report in Table~\ref{tab:selection} the top 12 important variables.
	
	Among 94 environmental variables, the AL and ALDS respectively selects 32 and 33 important covariates (most of them are similar). We sort the magnitude of $\hat \bbeta$ to identify the 12 most informative covariates. It turns out that these 12 covariates are similar for both procedures (see Table~\ref{tab:selection}). For the rest of selected covariates, the rankings are slightly different but the magnitudes are very similar (see Table~\ref{tab:betaest}).
	
	\setlength{\tabcolsep}{2.5pt}
	\renewcommand{\arraystretch}{1}
	\begin{table}[ht]
		\caption{Twelve most important covariates selected by AL and ALDS for modeling the intensity of Acalypha diversifolia point pattern}
		\label{tab:selection}
		\centering
		\begin{tabular}{lrr}
			\hline
			Covariates & AL & ALDS \\ 
			\hline
			Ca:N.min & -0.89 & -0.89 \\ 
			K:N.min & 0.58 & 0.53 \\ 
			Al:Mg & -0.51 & -0.47 \\ 
			pH & 0.48 & 0.46 \\ 
			B & -0.46 & -0.43 \\ 
			Ca & 0.45 & 0.42 \\ 
			Al:Fe & 0.38 & 0.39 \\ 
			Fe:K & -0.31 & -0.33 \\ 
			B:P & -0.30 & -0.29 \\ 
			P:Nz & 0.27 & 0.26 \\ 
			Fe & 0.26 & 0.25 \\
			Mn & -0.24 & -0.24 \\
			Number of selected covariates & 32 & 33 \\ 
			\hline
		\end{tabular}
	\end{table}

	\section{Discussion}\label{sec:conl}
	
	In this paper, we develop the adaptive lasso and Dantzig selector for spatial point processes intensity estimation and provide asymptotic results under an original setting where the number of non-zero and zero coefficients diverge with the mean number of points. We demonstrate that both methods share identical asymptotic properties and perform similarly on simulated and real data. This study supplements previous ones \citep[see e.g.][]{bickel2009simultaneous} where similar conclusions on linear models and generalized linear models were addressed.
	
	\cite{choiruddin2018convex} considered extensions of lasso type methods by involving general convex and non-convex penalties. In particular, composite likelihoods penalized by SCAD or MC+ penalty showed interesting properties. To integrate such an idea for the Dantzig selector, we could consider the  optimization problem
	\begin{align*}
		\min \sum_{i=1}^{p}{p_{\lambda_j}(\beta_j)} \mbox{ subject to }  {\mu}^{-1}\; \Big |\mathbf{U}_j( \tilde {\bbeta}) +  [\mathbf{A}( \tilde {\bbeta}) (\tilde {\bbeta}  -  {\bbeta})]_j \Big | \leq  p^{\prime}_{\lambda_j}(\beta_j), && j=1,\ldots,p
	\end{align*}
	where $p^{\prime}_{\lambda}(\theta)$ is the derivative with respect to $\theta$ of a general penalty function $p_\lambda$. However, such an interesting extension would make linear programming unusable and theoretical developments more complex to derive.
	We leave this direction for further study.
	
	{Another direction for further study is to derive results for the selection of regularization parameters. As mentioned earlier, a challenging and definitely interesting perspective would be the validity of Theorem~\ref{thm:main} when we define the regularization parameters in a stochastic way such as $\lambda_{n,j}=\lambda_n/|\tilde \beta_j|$.}

	On a similar topic, \cite{choiruddin2021information} studied information criteria such as AIC, BIC, and the composite-versions under similar asymptotic framework for selecting intensity model of spatial point process. These criteria could be extended for tuning parameter selection in the context of regularization methods for spatial point process.

	
	\begin{appendix}
		
		
		\section{Additional notation and auxiliary Lemmas}  \label{auxLemma}
		
		Lemmas~\ref{bound}-\ref{lem:rhoubeta} are used for the proof of Theorem~\ref{thm:main} in both cases $\hat \bbeta=\hat \bbeta_{\mathrm{AL}},\hat \bbeta_{\mathrm{ALDS}}$. Throughout the proofs, the notation $\mathbf X_n = O_{\mathrm P} (x_n)$ or $\mathbf X_n = o_{\mathrm P} (x_n)$ for a random vector $\mathbf X_n$ and a sequence of real numbers $x_n$ means that $\|\mathbf X_n\|=O_{\mathrm P}(x_n)$ and $\|\mathbf X_n\|=o_{\mathrm P}(x_n)$. In the same way for a vector $\mathbf V_n$ or a squared matrix $\mathbf M_n$, the notation $\mathbf V_n=O(x_n)$ and  $\mathbf M_n=O(x_n)$ mean that $\|\mathbf V_n\|=O(x_n)$ and  $\|\mathbf M_n\|=O(x_n)$. 
		
		\begin{lemma} \label{bound}
			Under conditions~\cond{C:nun}, \cond{C:cov} and~\cond{C:g}, the following result hold as $n\to \infty$
			\begin{align}
				\max\left\{ \|\mathbf{U}_n(\bbeta_0)\|,\|\mathbf{U}_{n,2}(\bbeta_0)\|\right\}  =O_\mathrm{P} \left( \sqrt { p_n \mu_n} \right) 
				\quad \text{ and } \quad
				\mathbf{U}_{n,1}(\bbeta_0)  =O_\mathrm{P} \left( \sqrt { {s_n \mu_n} } \right).
				\label{ch2:eq:ln}
			\end{align}
		\end{lemma}
		
		\begin{proof}
			Using Campbell Theorems~\eqref{eq:campbell}, the score vector  $\mathbf{U}_n(\bbeta_0)$ is unbiased and has variance $\mathrm{Var} \mathbf{U}_{n}( \bbeta_0)= \mathbf{B}_{n}( \bbeta_{0})$. {By Condition~\cond{C:cov} for any $u\in D_n$, $\bz(u)\bz(u)^\top= O(p_n)$. Hence, $\mathbf A_n(\bbeta_0) = O(p_n \mu_n)$. By definition of~\eqref{eq:Bn} and conditions~\cond{C:nun} and~\cond{C:g}, we deduce that $\mathbf B_n(\bbeta_0)= O(p_n \mu_n)$. We deduce that $\mathrm{Var}\{\mathbf U_n(\bbeta_0)\}= O(p_n \mu_n)$. In the same way, $\mathrm{Var}\{\mathbf U_{n,2}(\bbeta_0)\}= O\{(p_n-s_n) \mu_n\} = O(p_n \mu_n)$ and $\mathrm{Var}\{\mathbf U_{n,1}(\bbeta_0)\}= O(s_n \mu_n)$. } The result is proved since for any centered real-valued stochastic process ${Y_n}$ with finite variance $\mathrm{Var}({Y_n})$, ${Y_n}=O_\mathrm{P}{\big\{\sqrt{\mathrm{Var} ({Y_n})}\big\}}$.
		\end{proof}
		
		The next lemma states  that in the vicinity of $\bbeta_0$, $\rho(u;\bbeta)$ and $\rho(u;\bbeta_0)$ have the same behaviour.
		\begin{lemma}\label{lem:rhoubeta}
			(i) Let $(\zeta_n)_{n\ge 1}$ be any sequence such that $\zeta_n=o(1/\sqrt{p_n})$ and let $\kappa$ be any non-negative real number, then under the Conditions~\cond{C:intensity}-\cond{C:cov}, we have 
			\[
			\sup_{\|\bbeta - \bbeta_0\| \le \kappa \zeta_n} \int_{D_n} \rho(u;\bbeta) \dd u = O(\mu_n).
			\]
			(ii) Similary for any random vector $\bbeta$ such that $\|\bbeta-\bbeta_0\|=o_\P(1/\sqrt{p_n})$, then
			\[
			\int_{D_n} \rho(u;\bbeta) \dd u = O_\P(\mu_n).	
			\]
			(iii) In addition, under condition~\cond{C:snpn}, (i)-(ii) are valid for the sequence defined by $\zeta_n= \sqrt{p_n/ \mu_n}$.
		\end{lemma}
		
		\begin{proof}
			(i)-(ii) We only focus on (i) as (ii) follows along similar lines. For any $u \in D_n$, there exists, by Condition~\cond{C:cov}, a constant $\kappa<\infty$ (independent of $u, \bbeta$ and $\bbeta_0$) such that 
			\[
			- \kappa \sqrt{p_n}\|\bbeta - \bbeta_0\| \leq  (\bbeta-\bbeta_0)^\top \bz (u) \leq 	\kappa \sqrt{p_n}\|\bbeta - \bbeta_0\|.	
			\]
			Since, $\int \rho(u;\bbeta)\dd u = \int \exp\{(\bbeta-\bbeta_0)^\top \bz (u)\} \rho(u;\bbeta_0) \dd u$, we deduce that
			\[
			\mu_n \, \exp(- \kappa \sqrt{p_n}\|\bbeta - \bbeta_0\|)	\le	\int_{D_n} \rho(u;\bbeta) \dd u \le \mu_n \, \exp( \kappa \sqrt{p_n}\|\bbeta - \bbeta_0\|) 
			\]	
			which yields the result by definition of $\zeta_n$.\\
			(iii) Condition~\cond{C:snpn} implies in particular that $\sqrt{p_n^2/\mu_n} \to 0$ as $n \to \infty$.
		\end{proof}

		\section{Proof of Theorem~\ref{thm:main} when $\hat \bbeta=\hat \bbeta_{\mathrm{AL}}$} \label{sec:proofAL}
		
		\subsection{Existence of a root-$(\mu_n/p_n)$ consistent local maximizer}
		
		The first result presented hereafter shows that there exists a local maximizer of $Q_n(\bbeta)$ which is a consistent estimator of $\bbeta_0$.
		
		\begin{proposition}
			\label{proposition:AL}
			Assume that the conditions~\cond{C:nun}-\cond{C:g} hold. If in addition $p_n^4/\mu_n\to 0$ and $a_n\sqrt{s_n \mu_n/p_n}\to 0$ as $n\to \infty$, then there exists a local maximizer ${\hat \bbeta_{\mathrm{AL}}}$ of $Q_n(\bbeta)$  such that
			\begin{align*}
				{\bf \| \hat \bbeta_{\mathrm{AL}} -\bbeta_0\|}=O_\mathrm{P}\big\{ \sqrt{p_n/\mu_n}  \big\}.
			\end{align*}
		\end{proposition}
		
		Note that the conditions on $a_n, s_n, \mu_n$ and $p_n$ are actually implied by conditions~\cond{C:snpn} and~\cond{C:anbn}.
		In the proof of this result and the following ones, the notation $\kappa$ stands for a generic constant which may vary from line to line. In particular this constant is independent of $n$, $\bbeta_0$ and $\mathbf k$.
		
		\begin{proof}
			Let $\mathbf{k}\in \mathbb{R}^{p_n}$. We remind the reader that the estimate of $\bbeta_0$ is defined as the maximum of the function $Q_n$, given by~\eqref{regmed},  over $\mathbb R^{p_n}$. 
			To prove Proposition~\ref{proposition:AL}, we aim at proving that for any given $\epsilon>0$, there exists sufficiently large $K>0$ such that for $n$ sufficiently large 
			\begin{equation}
				\label{ch2:eq:15}
				\mathrm{P}\bigg\{\sup_{\|\mathbf{k}\| = K} \Delta_n(\mathbf k)>0\bigg\}\leq \epsilon,
				\quad \mbox{ where } \Delta_n(\mathbf k) = Q_n(\bbeta_0+\sqrt{p_n/\mu_n}\mathbf{k})-Q_n(\bbeta_0).
			\end{equation}
			Equation~\eqref{ch2:eq:15} will imply that with probability at least $1-\epsilon$, there exists a local maximum in the ball $\{\bbeta_0+\sqrt{p_n/\mu_n}\mathbf{k}:\|\mathbf{k}\| \leq K\}$, and therefore  a local maximizer $\boldsymbol{\hat{\beta}}$  such that $\|{ \boldsymbol {\hat \beta}-\bbeta_0}\|=O_\mathrm{P}(\sqrt{p_n/\mu_n})$.  We decompose $\Delta_n(\mathbf k)$ as $\Delta_n(\mathbf k)= T_1+T_2$ where
			\begin{align*}
				T_1 & = \mu_n^{-1} \left\{ \ell_n(\bbeta_0+\sqrt{p_n/\mu_n}\mathbf{k})-\ell_n( \bbeta_0) \right\} \\
				T_2 & = \sum_{j=1}^{p_n} \lambda_{n,j} \left( |\beta_{0j}|- |\beta_{0j}+\sqrt{p_n/\mu_n}k_j| \right).
			\end{align*}
			Since $\rho(u;\cdot)$ is infinitely continuously differentiable and $\ell_n^{(2)}(\bbeta) =-\mathbf A_n(\bbeta)$, then using a second-order Taylor expansion there exists $t\in (0,1)$ such that
			\begin{align*}
				\mu_n T_1 =& \, \sqrt{p_n/\mu_n} \mathbf k^\top \ell_n^{(1)}(\bbeta_0) + T_{11}+T_{12} \\
			\end{align*}
			where
			\begin{align*}
				T_{11} =&- \frac12 \frac{p_n}{\mu_n}\mathbf k^\top \mathbf{A}_n(\bbeta_0) \mathbf k  \\
				T_{12}=&+ \frac12\frac{p_n}{\mu_n}\mathbf k^\top \left\{ \mathbf{A}_n(\bbeta_0) -\mathbf{A}_n(\bbeta_0 + t\sqrt{p_n/\mu_n} \mathbf k) \right\} \mathbf k .	
			\end{align*}
			By condition~\cond{C:cov}
			\[
			T_{11} = -\frac12 p_n \frac{\mathbf k^\top \{\mu_n^{-1}\mathbf A_n(\bbeta_0)\} \mathbf k}{\|\mathbf k\|^2} \, \|\mathbf k\|^2 \le -\frac{\alpha}2 p_n \|\mathbf k\|^2
			\] 
			where $\alpha = \liminf_{n\ge 1}\inf_{\mathbf \boldsymbol \phi, \|\mathbf \boldsymbol \phi\|=1} \boldsymbol \phi^\top \{\mu_n^{-1}\mathbf A_n(\bbeta_0)\} \boldsymbol \phi>0$. Now, for some $\tilde \bbeta$ on the line segment between $\bbeta_0$ and $\bbeta_0 + t{\sqrt{p_n/\mu_n}} \mathbf k$
			\[
			T_{12} = \frac12 \, \frac{p_n}{\mu_n} \mathbf k^\top 
			\left\{
			\int_{D_n} \bz(u) \bz(u)^\top t \sqrt{\frac{p_n}{\mu_n}} \mathbf k^\top \bz(u) \rho(u;\tilde \bbeta) \dd u
			\right\} 
			\mathbf k.
			\]
			By conditions~\cond{C:nun}-\cond{C:cov} and Lemma~\ref{lem:rhoubeta}
			\[
			T_{12} = O\left( \|\mathbf k\|^3 \frac{p_n}{\mu_n} p_n \sqrt{\frac{p_n}{\mu_n}} \sqrt{p_n} \mu_n\right) = O\left(p_n \sqrt{\frac{p_n^4}{\mu_n}} \right) = o(p_n).
			\]
			Hence, for $n$ sufficiently large
			\[
			\mu_n T_1 \le \sqrt{\frac{p_n}{\mu_n}} \mathbf k^\top \ell_n^{(1)}(\bbeta_0) - \frac{\alpha}4 p_n \|\mathbf k\|^2.
			\]
			Regarding the term $T_2$ we have,
			\[
			T_2\leq \sum_{j=1}^{s_n} \lambda_{n,j} 
			\left\{ 
			|\beta_{0j}|- \left|\beta_{0j}+ \sqrt{\frac{p_n}{\mu_n}} k_j\right|
			\right\} 
			\leq a_n \sqrt{\frac{p_n}{\mu_n}} \sum_{j=1}^{s_n} |k_j| \leq a_n \sqrt{\frac{s_np_n}{\mu_n}} \|\mathbf k\|.
			\]
			We deduce that for $n$ large enough, there exists $\kappa$ such that $T_2 \leq \kappa d_n^2 \|k\|$ whereby we deduce that
			\[
			\Delta_n(\mathbf k) \leq 
			\frac{1}{\mu_n}\sqrt{\frac{p_n}{\mu_n}} \mathbf k^\top \ell_n^{(1)}(\bbeta_0)  
			-\frac{\alpha}4 \frac{p_n}{\mu_n}  \|\mathbf k\|^2 + a_n \sqrt{\frac{s_np_n}{\mu_n}} \|\mathbf k\|.
			\]
			By the assumption of Proposition~\ref{proposition:AL}, $a_n\sqrt{s_np_n/\mu_n}= a_n \sqrt{s_n\mu_n/p_n} p_n/\mu_n = o(p_n/\mu_n)$, whereby we deduce that for $n$ sufficiently large
			\[
			\Delta_n(\mathbf k) \le \frac1{\mu_n} \sqrt{\frac{p_n}{\mu_n}} \mathbf k^\top \ell_n^{(1)}(\bbeta_0)  
			-\frac{\alpha}8 \frac{p_n}{\mu_n}  \|\mathbf k\|^2.
			\]
			Now for $n$ sufficiently large,
			\begin{align*}
				\mathrm{P}\bigg\{{\sup_{\|\mathbf{k}\|= K}  \Delta_n(\mathbf{k})>0}\bigg\} &\leq 
				\mathrm{P}\bigg\{ \|\ell_n^{(1)}(\bbeta_0)\| \ge \frac\alpha8 K p_n \sqrt{\frac{\mu_n}{p_n}} \bigg\} \\
				&= \mathrm{P}\bigg\{ \|\ell_n^{(1)}(\bbeta_0)\| \ge \frac\alpha8 K \sqrt{p_n \mu_n} \bigg\}<\varepsilon
			\end{align*}
			for any given $\varepsilon>0$ since $\ell_n^{(1)}(\bbeta_0) = \mathbf U_{n}(\bbeta_0)= O_\P(\sqrt{p_n\mu_n})$ by Lemma~\ref{bound}.
			
		\end{proof}

		\subsection{Sparsity property for $\hat \bbeta=\hat \bbeta_{\mathrm{AL}}$}
		
		The sparsity property for $\hat \bbeta_{\mathrm{AL}}$ follows directly from Proposition~\ref{proposition:AL} and the following Lemma~\ref{sparsity}.
		
		\begin{lemma}
			\label{sparsity}
			Assume the conditions~\cond{C:nun}-\cond{C:g} and~\cond{C:snpn}-\cond{C:anbn} hold, then with probability tending to $1$, for any {$\bbeta_1 \in \mathbb{R}^{s_n}$} satisfying $\|{\bbeta_1 - \bbeta_{01}}\|=O_\mathrm{P}(\sqrt{p_n/\mu_n})$, and for any constant $K_1 > 0$,
			\begin{align*}
				Q_n\Big\{({\bbeta_1}^\top,\mathbf{0}^\top)^\top \Big\}
				= \max_{\| \bbeta_2\| \leq K_1 \sqrt{p_n/\mu_n}}
				Q_n\Big\{({\bbeta_1}^\top,{\bbeta_2}^\top)^\top \Big\}.
			\end{align*}
		\end{lemma}
		\begin{proof}
			Let $\varepsilon_n= K_1 \sqrt{p_n/\mu_n}$. It is sufficient to show that with probability tending to $1$ as ${n\to \infty}$, for any ${\bbeta_1}$ satisfying $\|{\bbeta_1 -\bbeta_{01}}\|=O_\mathrm{P}(\sqrt{p_n/\mu_n})$, we have for any $j=s_n+1, \ldots, p_n$
			
			\begin{equation}
				\label{sparsitya}
				\frac {\partial Q_n(\bbeta)}{\partial\beta_j}<0 \quad
				\mbox { for } 0<\beta_j<\varepsilon_n, \mbox{ and}
			\end{equation}
			
			\begin{equation}
				\label{sparsityb}
				\frac {\partial Q_n(\bf \bbeta)}{\partial\beta_j}>0 \quad
				\mbox { for } -\varepsilon_n<\beta_j<0.
			\end{equation}
			From \eqref{eq:likepois},
			\begin{align*}
				\frac {\partial \ell_n(\bbeta)}{\partial\beta_j} = \frac {\partial \ell_n{(\bbeta_0)}}{\partial\beta_j} + R_n,
			\end{align*}
			where $R_n =  {-} \int_{D_n}  z_j(u)\big\{\rho(u;\bbeta)-\rho(u;\bbeta_0)\big\} \mathrm{d}u$. 
			Let $u \in \mathbb{R}^d$. By Taylor expansion, there exists $ t\in (0,1), $ such that 
			\begin{align*}
				\rho(u;\bbeta) = \rho(u;\bbeta_0) + (\bbeta-\bbeta_0)^\top \bz(u) \rho\{u;\bbeta_0 + t(\bbeta-\bbeta_0 )\}.
			\end{align*}
			By condition~\cond{C:intensity}-\cond{C:cov} and Lemma~\ref{lem:rhoubeta}, we have for $n$ sufficiently large
			\[
			|R_n| \le \kappa \|\bbeta-\bbeta_0\| \sqrt{p_n} \int_{D_n} \rho(u;\bbeta_0) \dd u = 
			O_{\mathrm P}\left( \sqrt{\frac{p_n}{\mu_n}} \sqrt{p_n} \mu_n \right) =
			O_{\mathrm P}\left( {p_n} \sqrt{\mu_n} \right). 
			\]
			Following the proof of Lemma~\ref{bound}, we can derive $\mathrm{Var}({\partial \ell_n{(\bbeta_0)}}/{\partial\beta_j} ) = \mathrm{Var}[\{\mathbf U_n(\bbeta_0) \}_j ]=O(\mu_n)$ whereby we deduce that
			\begin{equation}
				\label{Op}
				\frac {\partial \ell_n(\bbeta)}{\partial\beta_j} = O_{\mathrm P} (p_n \sqrt{\mu_n}).
			\end{equation}
			
			Now, we want to prove (\ref{sparsitya}). Let $0<\beta_j<\varepsilon_n$ and remind that the sequence $b_n$ is given by~\cond{C:anbn}. Then, for $n$ sufficiently large,
			\begin{align*}
				\mathrm{P} \left\{ \frac {\partial Q_n(\bbeta)}{\partial\beta_j}<0 \right\}&=\mathrm{P} \left\{ \frac {\partial \ell_n(\bbeta)}{\partial\beta_j} - \mu_n{\lambda_{n,j}}\sign(\beta_j)<0 \right\}\\
				&=\mathrm{P} \left\{ \frac {\partial \ell_n(\bbeta)}{\partial\beta_j}< \mu_n{\lambda_{n,j}} \right\}\\
				& \geq \mathrm{P} \left\{ \frac {\partial \ell_n(\bbeta)}{\partial\beta_j}< \mu_n b_n \right\}\\
				&= \mathrm{P} \left\{ \frac {\partial \ell_n(\bbeta)}{\partial\beta_j}< p_n\sqrt{\mu_n} \; \sqrt{\frac{\mu_n}{p_n^2}}b_n \right\}.
			\end{align*}
			The assertion (\ref{sparsitya}) is therefore deduced from (\ref{Op}) and from the assumption that $b_n \sqrt{\mu_n/p_n^2} \to \infty$ as $n  \to \infty$. We proceed similarly to prove (\ref{sparsityb}).
		\end{proof}
		
		\subsection{Asymptotic normality for $\hat \bbeta=\hat \bbeta_{\mathrm{AL}}$}

		\begin{proof}
			As shown in Proposition~\ref{proposition:AL}, there is a root-$(\mu_n/p_n)$ consistent local maximizer $\hat \bbeta_{\mathrm{AL}}$ of $Q_n(\bbeta)$, and it can be shown that there exists an estimator $\hat\bbeta_{\mathrm{AL},1}$ in Proposition~\ref{proposition:AL} that is a root-$(\mu_n/p_n)$ consistent local maximizer of $ Q_n \{({\bbeta_1}^\top,\mathbf{0}^\top)^\top \Big\}$, which is regarded as a function of  $\boldsymbol {\beta}_1$, and that satisfies
			\begin{align*}
				\frac {\partial Q_n(\hat\bbeta_{\mathrm{AL}})}{\partial\beta_j}=0 \quad
				\mbox { for } j=1,\ldots,s_n \mbox { and } \hat\bbeta_{\mathrm{AL}}=( \hat\bbeta_{\mathrm{AL},1}^\top,\mathbf{0}^ \top)^\top.
			\end{align*}
			There exists $t\in (0,1)$ and $\boldsymbol{\check{\beta}}= \hat\bbeta_{\mathrm{AL}} + t(\bbeta_0-\hat\bbeta_{\mathrm{AL}})$  such that for $j=1,\cdots,s_n$
			\begin{align}
				0
				=&\frac {\partial \ell_n{(\hat\bbeta_{\mathrm{AL}})}}{\partial\beta_j}-\mu_n{\lambda_{n,j}}\sign({\hat \beta_{\mathrm{AL},j}}) \nonumber\\
				=&\frac {\partial \ell_n{(\bbeta_0)}}{\partial\beta_j}+{\sum_{l=1}^{s_n} \frac {\partial^2 \ell_n{( \boldsymbol{\check{\beta}})}}{\partial\beta_j \partial\beta_l}}({\hat \beta_{\mathrm{AL},l}-\beta_{0l}})-\mu_n{\lambda_{n,j}}\sign({\hat \beta_{\mathrm{AL},j}}) \nonumber\\
				=&\frac {\partial \ell_n{(\bbeta_0)}}{\partial\beta_j}+{\sum_{l=1}^{s_n} \frac {\partial^2 \ell_n{( \bbeta_0)}}{\partial\beta_j \partial\beta_l}}({\hat \beta_{\mathrm{AL},l}-\beta_{0l}})+{\sum_{l=1}^{s_n} \Psi_{n,jl}({\hat \beta_{\mathrm{AL},l}}-\beta_{0l})} \nonumber \\
				&-\mu_n\lambda_{n,j}\sign(  {\hat \beta_{\mathrm{AL},j}}   )
				\label{eq:0equal}
			\end{align}
			where 
			\begin{align*}
				\Psi_{n,jl}=\frac {\partial^2 \ell_n{(\boldsymbol{\check{\beta}})}}{\partial\beta_j \partial\beta_l}-\frac {\partial^2 \ell_n{(\bbeta_0)}}{\partial\beta_j \partial\beta_l}.
			\end{align*}
			Let $\mathbf{U}_{n,1}(\bbeta_{0})$ (resp. $\ell^{(2)}_{n,1}(\bbeta_{0})$) be the first $s_n$ components (resp. $s_n \times s_n$ top-left corner) of $\mathbf{U}_{n}(\bbeta_{0})$ (resp. $\ell^{(2)}_{n}(\bbeta_{0})$). Let also $\boldsymbol \Psi_n$ be the $s_n \times s_n$ matrix containing $\Psi_{n,jl}, j,l=1,\ldots,s_n$. Finally, let the vector $\mathbf{p}'_n$ 
			\begin{align*}
				\mathbf{p}'_n&={ 
					\{\lambda_{n,1}\sign({\hat \beta_{\mathrm{AL},1}} ),\ldots,
					\lambda_{n,s_n}\sign( {\hat \beta_{\mathrm{AL},s_n}} )\}^\top}.  
			\end{align*}
			These notation allow us to rewrite~\eqref{eq:0equal} as
			\begin{equation}
				\label{eq:tmp}
				\mathbf U_{n,1}(\bbeta_0) - \mathbf A_{n,11}(\bbeta_0) (\hat \bbeta_{\mathrm{AL},1}-\bbeta_{01}) +
				\boldsymbol \Psi_n (\hat \bbeta_{\mathrm{AL},1}-\bbeta_{01}) - \mu_n \mathbf{p}'_n =0.
			\end{equation}
			Let $\boldsymbol \phi \in \mathbb R^{s_n}\setminus \{0\}$ and $\sigma^2_\phi = \boldsymbol \phi^\top \mathbf B_{n,11}(\bbeta_0) \boldsymbol \phi$, then
			\[
			\sigma_{\boldsymbol \phi}^{-1} \boldsymbol \phi^\top\mathbf U_{n,1}(\bbeta_0) - \sigma_{\boldsymbol \phi}^{-1} \boldsymbol \phi^\top\mathbf A_{n,11}(\bbeta_0) (\hat \bbeta_{\mathrm{AL},1}-\bbeta_{01}) +
			\sigma_{\boldsymbol \phi}^{-1} \boldsymbol \phi^\top\boldsymbol \Psi_n (\hat \bbeta_{\mathrm{AL},1}-\bbeta_{01}) - \mu_n \sigma_{\boldsymbol \phi}^{-1} \boldsymbol \phi^\top\mathbf{p}'_n =0.
			\]
			Now, by condition~\cond{C:Bn}, $\sigma_{\boldsymbol \phi}^{-1} = O(\mu_n^{-1/2})$ and by the definition of $a_n$, $\mathbf p'_n= O(a_n \sqrt{s_n})$. By conditions~\cond{C:intensity}-\cond{C:cov}, there exists some $\tilde \bbeta$ on the line segment between $\bbeta_0$ and $\check \bbeta$ such that
			\[
			\boldsymbol \Psi_n = \int_{D_n} \bz_1(u) \bz_1(u)^\top (\check \bbeta-\bbeta_0)^\top \bz(u) \rho(u;\tilde \bbeta) \dd u
			\]
			whereby we deduce from conditions~\cond{C:intensity}-\cond{C:cov} and Lemma~\ref{lem:rhoubeta} that
			\[
			\|\boldsymbol \Psi_n\| = O_{\mathrm P} \left( s_n \sqrt{\frac{p_n}{\mu_n}} \sqrt{p_n} \mu_n\right) = O_{\mathrm P}(s_np_n\sqrt{\mu_n}).
			\]
			The last two results and conditions~\cond{C:snpn}-\cond{C:anbn} yield that
			\begin{align*}
				\sigma_{\boldsymbol \phi}^{-1} \boldsymbol \phi^\top\boldsymbol \Psi_n (\hat \bbeta_{\mathrm{AL},1}-\bbeta_{01})&= 
				O_{\mathrm P} \left( \frac1{\sqrt{\mu_n}} s_np_n\sqrt{\mu_n} \sqrt{\frac{p_n}{\mu_n}} \right) = 
				O_{\mathrm P} \left( \sqrt{\frac{s_n^2 p_n^3}{\mu_n}}\right) = o_{\mathrm P}(1)\\
				\mu_n \sigma_{\boldsymbol \phi}^{-1} \boldsymbol \phi^\top\mathbf{p}'_n &=
				O \left(\mu_n \frac{1}{\sqrt{\mu_n}} a_n\sqrt{s_n}  \right) = O(a_n\sqrt{s_n\mu_n}) = o(1).
			\end{align*}
			These results finally lead to
			\[
			\sigma_{\boldsymbol \phi}^{-1} \boldsymbol \phi^\top\mathbf A_{n,11}(\bbeta_0) (\hat \bbeta_{\mathrm{AL},1}-\bbeta_{01})  = \sigma_{\boldsymbol \phi}^{-1} \boldsymbol \phi^\top\mathbf U_{n,1}(\bbeta_0)  + o_{\mathrm P}(1)
			\]
			and finally to the proof of the result using Slutsky's lemma and condition~\cond{C:clt}.

		\end{proof}

		\section{Proof of Theorem~\ref{thm:main} when $\hat \bbeta=\hat \bbeta_{\mathrm{ALDS}}$} \label{sec:proofALDS}

		\subsection{Existence and optimal solutions for  the primal and dual problems} \label{sec:existenceALDS}

		For $\bbeta \in \mathbb R^{p_n}$, we let $\boldsymbol\Delta_n (\bbeta) = \mathbf{U}_{n}( \tilde {\bbeta}) +  \mathbf{A}_n( \tilde {\bbeta}) (\tilde {\bbeta} -\bbeta )$.

		\begin{lemma} \label{lem:existenceALDS}
			There exists a solution to the problem~\eqref{ADS2}.
		\end{lemma}

		\begin{proof}
			Following \cite{candes:romberg:05}, we state that \eqref{ADS2} is equivalent to
			\begin{align}
				\label{ADS-linear}
				\min_{\bbeta, {\boldsymbol u} } \sum_j u_j \mbox{ subject to }  \begin{cases}
					\bld_{n} \bbeta \leq {\boldsymbol u} \\
					-\bld_{n} \bbeta \leq {\boldsymbol u} \\
					\mu_n^{-1}\bld_{n}^{-1} \boldsymbol \Delta_n(\bbeta) - \boldsymbol 1_{p_n} \leq {\mathbf{0}} \\
					- \mu_n^{-1} \bld_{n}^{-1} \boldsymbol \Delta_n(\bbeta) - \boldsymbol 1_{p_n} \leq {\mathbf{0} }
				\end{cases}
			\end{align}
			where $\boldsymbol u  \in \mathbb{R}^{p_n}$ is an additional parameter vector to be optimized and $\tilde {\bbeta}$ is the initial estimator. Note that~\eqref{ADS-linear} is a linear problem with $4p_n$ linear inequality constraints. To prove the existence of ALDS estimates, we need to derive dual problem of \eqref{ADS-linear} and prove that strong duality holds. To derive the dual problem, we first construct the Lagrangian form associated with the problem \eqref{ADS-linear} considering the main arguments by \cite[section 5.2]{boyd2004convex}
			\begin{align*}
				L(\bbeta;  \boldsymbol u ; \baf) & = \sum_j u_j  
				+ \baf_1^\top (\bld_{n} \bbeta -\mathbf u)
				+ \baf_2^\top (-\bld_{n} \bbeta -\mathbf u)\notag\\
				&  \quad + \baf_3^\top \Big[ \mu_n^{-1} \bld_{n}^{-1} 
				\boldsymbol \Delta_n(\bbeta)- \boldsymbol 1_{p_n} \Big]  + \baf_4^\top \Big[ - \mu_n^{-1} \bld_{n}^{-1} \boldsymbol \Delta_n(\bbeta) - \boldsymbol 1_{p_n} \Big]\\
				& =  (\boldsymbol 1_{p_n} - \baf_1 - \baf_2)^\top \boldsymbol u + \Big \{  (\baf_1-\baf_2)^\top\bld_{n}
				- \mu_n^{-1} (\baf_3-\baf_4)^\top\bld_{n}^{-1} \mathbf{A}_n( \tilde {\bbeta})
				\Big \}\bbeta \notag \\
				&  \quad + (\baf_3-\baf_4)^\top \Big \{ \mu_n^{-1} \bld_{n}^{-1} \Big(\mathbf{U}_{n}( \tilde {\bbeta}) +  \mathbf{A}_n( \tilde {\bbeta}) \tilde {\bbeta} \Big)  \Big \} - (\baf_3+\baf_4)^\top \boldsymbol 1_{p_n},
			\end{align*}
			where $\baf =(\baf_1^\top, \baf_2^\top, \baf_3^\top, \baf_4^\top)^\top \in \mathbb{R}^{4p_n}$ is the dual vector (which can be viewed as a Lagrange multiplier). 
			
			The dual function  $h$ is defined by
			\begin{align}
				h(\baf)&=\inf_{\bbeta, \boldsymbol u} L(\bbeta; \boldsymbol u; \baf) \\
				&=
				\begin{cases}
					(\baf_3-\baf_4)^\top \Big \{ \mu_n^{-1}\bld_{n}^{-1} \Big(\mathbf{U}_{n}( \tilde {\bbeta}) +  \mathbf{A}_n( \tilde {\bbeta}) \tilde {\bbeta} \Big)  \Big \} - (\baf_3+\baf_4)^\top \boldsymbol 1_{p_n} , \text{ if} \\
					\quad 
					\begin{cases}
						\boldsymbol 1_{p_n} - \baf_1 - \baf_2 ={\mathbf{0}} \\
						(\baf_1-\baf_2)^\top\bld_{n}
						- \mu_n^{-1} (\baf_3-\baf_4)^\top\bld_{n}^{-1} \mathbf{A}_n( \tilde {\bbeta}) ={\mathbf{0}}
					\end{cases}
					\\
					-\infty \text{ otherwise}.
				\end{cases} \nonumber
			\end{align}
			For any $\baf =(\baf_1^\top, \baf_2^\top, \baf_3^\top, \baf_4^\top)^\top \in \mathbb{R^+}^{4p_n}$, $h(\baf)$ is a lower bound to the optimality problem \eqref{ADS-linear} (see \cite[p.216]{boyd2004convex}). To find the best lower bound comes to solve the dual problem: $\max_{\baf \geq {\mathbf{0}}} h(\baf).$

			Recall that problem \eqref{ADS-linear} is a linear program with linear inequality constraints, so that strong duality holds if the dual problem is feasible \cite[see][p.227]{boyd2004convex}, that is to say if there exists some $\baf =(\baf_1^\top, \baf_2^\top, \baf_3^\top, \baf_4^\top)^\top \in \mathbb{R^+}^{4p_n}$ such that
			\begin{eqnarray}
				\boldsymbol 1_{p_n} - \baf_1 - \baf_2 &=& {\mathbf{0}} \notag \\
				(\baf_1-\baf_2)^\top\bld_{n}
				- \mu_n^{-1} (\baf_3-\baf_4)^\top\bld_{n}^{-1} \mathbf{A}_n( \tilde {\bbeta}) &=&{\mathbf{0}}.
				\notag
			\end{eqnarray}
			Moreover, we remark that
			\begin{eqnarray}
				\baf_1 \geq {\mathbf{0}} \notag , \; 	\baf_2\geq {\mathbf{0}} \notag , \; \baf_3 \geq {\mathbf{0}} \notag, \; \baf_4 &\geq& {\mathbf{0}} \notag \\
				\boldsymbol 1_{p_n} - \baf_1 - \baf_2 &=&{\mathbf{0}} \notag \\
				(\baf_1-\baf_2)^\top\bld_{n}
				- \mu_n^{-1} (\baf_3-\baf_4)^\top\bld_{n}^{-1} \mathbf{A}_n( \tilde {\bbeta}) &=&{\mathbf{0}}
				\notag
			\end{eqnarray}
			is equivalent to 
			\begin{eqnarray}
				\baf_1 \geq {\mathbf{0}} \notag , \; 
				\baf_2=\boldsymbol 1_n - \baf_1 \geq {\mathbf{0}} \notag ,\; 
				\baf_3 \geq {\mathbf{0}} \notag ,\;
				\baf_4 &\geq& {\mathbf{0}} \notag \\
				(2\baf_1-\boldsymbol 1_{p_n})^\top\bld_{n}  
				-  \mu_n^{-1} (\baf_3-\baf_4)^\top\bld_{n}^{-1} \mathbf{A}_n( \tilde {\bbeta}) &=&{\mathbf{0}}
				\notag
			\end{eqnarray}
			which is also equivalent to 
			\begin{eqnarray}
				\baf_1 = \frac{1}{2} \Big\{\boldsymbol 1_{p_n}  + \mu_n^{-1}  \bld_{n}^{-1} \mathbf{A}_n( \tilde {\bbeta }) \bld_{n}^{-1} (\baf_3 - \baf_4)   \Big\}&\geq& {\mathbf{0}} \notag \\
				\baf_2 =\boldsymbol 1_{p_n} - \baf_1 = \frac{1}{2} \Big\{\boldsymbol 1_{p_n}  - \mu_n^{-1}  \bld_{n}^{-1} \mathbf{A}_n( \tilde {\bbeta }) \bld_{n}^{-1} (\baf_3 - \baf_4) \Big\}&\geq& {\mathbf{0}} \notag \\
				\baf_3 \geq {\mathbf{0}} \notag,\; 
				\baf_4 &\geq& {\mathbf{0}} \notag.
			\end{eqnarray}
			This comes to the condition that there exists $(\baf_3^\top, \baf_4^\top)^\top \in \mathbb{R^+}^{2p_n}$ such that
			\begin{equation}
				\label{existence}
				\mu_n^{-1} \| \bld_{n}^{-1} \mathbf{A}_n( \tilde {\bbeta }) \bld_{n}^{-1} (\baf_3 - \baf_4) \|_{\infty} \leq 1 .
			\end{equation}
			Therefore, the dual problem associated with \eqref{ADS-linear} is
			\begin{align}
				\max_{\baf_3, \baf_4 \geq 0} (\baf_3-\baf_4)^\top \Big[ \mu_n^{-1} \bld_{n}^{-1} \Big\{\mathbf{U}_{n}( \tilde {\bbeta}) +  \mathbf{A}_n( \tilde {\bbeta}) \tilde {\bbeta} \Big\}  \Big] - (\baf_3+\baf_4)^\top \boldsymbol 1_{p_n}  \nonumber \\
				\text{subject to } \mu_n^{-1}  \| \bld_{n}^{-1} \mathbf{A}_n( \tilde {\bbeta }) \bld_{n}^{-1} (\baf_3 - \baf_4) \|_{\infty} \leq 1 . \label{eq:dual2}
			\end{align}
			
			The condition \eqref{existence} is always true as far as the matrix $\bld_{n}^{-1} \mathbf{A}_n( \tilde {\bbeta }) \bld_{n}^{-1}$ is non zero. Indeed, let $\mathbf y \in \mathbb{R}^{p_n}$ such that $\mathbf \{\bld_{n}^{-1} \mathbf{A}_n( \tilde {\bbeta }) \bld_{n}^{-1} y \}  \neq 0$.  Now define 
			\begin{eqnarray*}
				\alpha_{3j} &=& \frac{y_j}{\mu_n^{-1} \|\bld_{n}^{-1} \mathbf{A}_n( \tilde {\bbeta }) \bld_{n}^{-1} y \|_{\infty}} \mathbf 1(y_j>0), \\
				\alpha_{4j} &=& \frac{-y_j}{\mu_n^{-1} \| \bld_{n}^{-1} \mathbf{A}_n( \tilde {\bbeta }) \bld_{n}^{-1} y \|_{\infty}} \mathbf 1(y_j<0).
			\end{eqnarray*}
			Clearly, $(\baf_3^\top, \baf_4^\top)^\top \in \mathbb{R^+}^{2p_n}$, 
			\eqref{existence} is always verified. This ends the proof.
		\end{proof}
		
		Note that the dual problem~\eqref{eq:dual2} can be unequivocally reparameterized in terms of $\bgam=\baf_3-\baf_4$ as follows
		\begin{align}
			\max_{\bgam \in \mathbb{R}^{p_n}} \bgam^\top \Big[ \mu_n^{-1} \bld_{n}^{-1} \Big\{\mathbf{U}_{n}( \tilde {\bbeta}) +  \mathbf{A}_n( \tilde {\bbeta}) \tilde {\bbeta} \Big\}  \Big] - \|\bgam\|_1  \nonumber \\
			\text{subject to } \mu_n^{-1}  \|\bgam^\top \bld_{n}^{-1} \mathbf{A}_n( \tilde {\bbeta }) \bld_{n}^{-1} \|_{\infty} \leq 1 \label{eq:dual}
		\end{align}
		due to complementary slackness conditions. Now, we derive the following optimality conditions and obtain optimal primal and dual solutions.

		We derive the following optimality conditions ensuring the Karush-Kuhn-Tucker (KKT) conditions and thus obtain optimal primal and dual solutions.
		
		\begin{lemma}
			\label{lemma:opt}
			Consider the primal and dual problems defined by \eqref{ADS2} and \eqref{eq:dual}. 
			Suppose that the matrix $\bld_{n}^{-1} \mathbf{A}_n( \tilde {\bbeta }) \bld_{n}^{-1}$ is non zero and that ${\hat \bbeta}$ and $\hat{\bgam}$ verify
			\begin{align}
				\mu_n^{-1} \Big\|\bld_{n}^{-1} 
				\boldsymbol \Delta_n({\hat \bbeta})
				\Big\|_\infty & \leq 1  \label{eq:fea1} \\  
				\mu_n^{-1} \Big\|\hat \bgam^\top \bld_{n}^{-1} \mathbf{A}_n( \tilde {\bbeta}) \bld_{n}^{-1} \Big\|_\infty  & \leq 1 \label{eq:fea2} \\
				\mu_n^{-1}
				\hat \bgam^\top  \bld_{n}^{-1} \mathbf{A}_n( \tilde {\bbeta}) \hat \bbeta & = \|\bld_n \hat \bbeta \|_1 \label{eq:slack1} \\
				\mu_n^{-1} \hat \bgam ^\top \bld_{n}^{-1} \boldsymbol\Delta_n({\hat \bbeta})  & = \|\hat \bgam \|_1. \label{eq:slack2}
			\end{align}
			Then the Karush-Kuhn-Tucker (KKT) conditions for \eqref{ADS2} are fulfilled and ${\hat \bbeta}$ and $\hat\bgam$ are the optimal primal and dual solutions.
		\end{lemma}
		
		\begin{proof}

			We start by writing the Karush-Kuhn-Tucker (KKT) conditions for the problem~\eqref{ADS2}:
			\begin{align}
				\bld_{n} \bbeta &\leq {\boldsymbol u} \label{KKT1}  \\
				-\bld_{n} \bbeta&\leq {\boldsymbol u} \label{KKT2}  \\
				\mu_n^{-1}\bld_{n}^{-1} {\boldsymbol \Delta_n(\bbeta)} - \boldsymbol 1_{p_n}& \leq {\mathbf{0}} \label{KKT3}  \\
				- \mu_n^{-1} \bld_{n}^{-1} {\boldsymbol \Delta_n(\bbeta)} - \boldsymbol 1_{p_n} &\leq {\mathbf{0}} \label{KKT4}  \\
				\baf_1 \geq 0, \baf_2 &\geq 0,  \label{KKT5} \\
				\baf_3 \geq 0, \baf_4 &\geq 0,  \label{KKT6} \\
				\forall i \; \alpha_{1i}\{(\bld_{n} \bbeta)_i -u_i\}&=0 \label{KKT7}\\
				\forall i \; \alpha_{2i}\{-(\bld_{n} \bbeta)_i -u_i\}&=0 \label{KKT8}\\
				\forall i \; \alpha_{3i}[\mu_n^{-1}\{\bld_{n}^{-1} {\boldsymbol \Delta_n(\bbeta)}\}_i - 1] &=0 \label{KKT9}\\
				\forall i \; \alpha_{4i}[-\mu_n^{-1}\{{(}\bld_{n}^{-1} {\boldsymbol \Delta_n(\bbeta)}\}_i - 1] &=0 \label{KKT10}\\
				1-\baf_1 - \baf_2 &=0 \label{KKT11}\\
				(\baf_1 - \baf_2)^T \bld_{n} - \mu_n^{-1} (\baf_3-\baf_4)^T \bld_{n}^{-1}  \mathbf{A}_n( \tilde {\bbeta}) &=0  \label{KKT12}
			\end{align}

			Let ${\hat \bbeta}$ and $\hat\bgam$ satisfy~\eqref{eq:fea1}-\eqref{eq:slack2}. \eqref{KKT3} and \eqref{KKT4} are obviously satisfied under~\eqref{eq:fea1}. If one defines $\hat{\boldsymbol \alpha}_3$ and $\hat{\boldsymbol \alpha}_4$ such that $\hat{\alpha}_{3i}=(0,\hat{\gamma}_i)_+$ and $\hat{\alpha}_{4i}=(0,-\hat{\gamma}_i)_+$, then $\hat{\bgam}=\hat{\boldsymbol \alpha}_3-\hat{\boldsymbol \alpha}_4$ and~\eqref{KKT6} is satisfied.
			
			Now, we define
			\begin{align*}
				\hat{\baf}_1 &= \frac{1}{2} \{\mathbf 1_{p_n} + \mu_n^{-1}  \bld_{n}^{-1} \mathbf{A}_n( \tilde {\bbeta}) \bld_{n}^{-1}\} (\hat \baf_3 - \hat \baf_4) \\
				\hat{\baf}_2 &= \frac{1}{2} \{\mathbf 1_{p_n} - \mu_n^{-1}  \bld_{n}^{-1} \mathbf{A}_n( \tilde {\bbeta}) \bld_{n}^{-1}\} (\hat \baf_3 - \hat \baf_4)
			\end{align*}
			which ensures~\eqref{KKT11} and~\eqref{KKT12}. In addition, under~\eqref{eq:fea2} implies that~\eqref{KKT5} is also true.
			
			From the definition of $\boldsymbol \alpha_3$ and $\boldsymbol \alpha_4$, we rewrite~\eqref{eq:slack2} as
			\[
			\sum_i 
			\left(
			\hat{\alpha}_{3i}[\mu_n^{-1}\{\bld_{n}^{-1} {\boldsymbol \Delta_n(\hat{\bbeta})}\}_i - 1] 
			+
			\hat{\alpha}_{4i}[-\mu_n^{-1}\{\bld_{n}^{-1} {\boldsymbol \Delta_n(\hat{\bbeta})}\}_i - 1] 
			\right)
			=0.
			\]
			From~\eqref{KKT4}-\eqref{KKT5} each term in the above sum is the sum of two negative terms, whereby we dedude that~\eqref{KKT9}-\eqref{KKT10} necessarily hold.
			With a similar argument, by using~\eqref{eq:slack1}, we also deduce that~\eqref{KKT7}-\eqref{KKT8} are also true by setting in particular $u_i = |(\bld_n \hat \bbeta)_i|$. And that latter choice implies that~\eqref{KKT1}-\eqref{KKT2} are also satisfied.
		\end{proof}

		
		\subsection{A few auxiliary statements}
		
		Before tackling more specifically the proof of Theorem~\ref{thm:main} for the ALDS estimator, we present a few auxiliary results that will be used. 
		
		\begin{lemma} \label{lem:aux} Assume conditions \cond{C:intensity}-\cond{C:initial} hold.\\
			(i) 
			\begin{equation}\label{lambda12}
				\|\bld_{n,11} \|= a_n 
				, \qquad
				\|\bld_{n,22} ^{-1}\|  = \frac1{b_n}.
			\end{equation}
			(ii) For any $t\in [0,1]$ and $\check \bbeta = \bbeta_0 + t(\tilde \bbeta -\bbeta_0)$, we have
			\begin{align*} 
				\mathbf{A}_n( {\check \bbeta})&=O_\P\left(p_n \mu_n\right) \\ 
				\mathbf{A}_{n,1}(\check {\bbeta}) & =O_\P\left(\sqrt{p_n s_n} \mu_n\right) \\
				\mathbf{A}_{n,2}(\check {\bbeta})  &=O_\P\left(p_n  \mu_n\right)  \\
				\mathbf{A}_{n,11}(\check {\bbeta}) & =O_\P\left( s_n \mu_n\right)\\
				\mathbf{A}_{n,21}(\check {\bbeta})  &=O_\P\left(\sqrt{p_ns_n} \mu_n \right) \\
				\mathbf A_n (\check \bbeta) - \mathbf A_n (\tilde \bbeta)  &= 
				O_\P\left(p_n^2 \sqrt{\mu_n}\right) \\
				\mathbf A_{n,1} (\check \bbeta) - \mathbf A_{n,1} (\tilde \bbeta)  &= 
				O_\P\left(  \sqrt{s_np_n^3\mu_n} \right) \\
				\mathbf A_{n,11} (\check \bbeta) - \mathbf A_{n,11} (\tilde \bbeta)  &= 
				O_\P\left( {\sqrt{s_n^2p_n^2 \mu_n}}\right).
			\end{align*}
			(iii) 
			\begin{equation}
				\label{eq:Unbetatilde}
				\max \left\{ \|\mathbf U_{n}(\tilde \bbeta)\|, \|\mathbf U_{n,2}(\tilde \bbeta)\| \right\} = 
				O_\P ( \sqrt{p_n^3 \mu_n}).
			\end{equation}
		\end{lemma}
		
		\begin{proof}
			(i) follows from conditions on $a_n$ and $b_n$. \\
			(ii) follows from conditions \cond{C:intensity}-\cond{C:cov}, ($\mathcal{C}$.\ref{C:initial}) and Lemma~\ref{lem:rhoubeta}. We only prove the assertions for the matrices $\mathbf A_n(\tilde \bbeta)$ and $\mathbf A_n(\check \bbeta)-\mathbf A_n(\tilde \bbeta)$. The other cases follow along similar lines. First,
			\begin{align*}
				\| \mathbf A_n(\tilde \bbeta) \| \leq \int_{D_n} \|\bz(u)\|^2 \rho(u;\tilde \bbeta) \dd u = O_\P(p_n \mu_n).
			\end{align*}
			Second, using Taylor expansion, there exists $\bbeta^\prime$ on the segment between $\tilde \bbeta$ and $\check \bbeta$ such that $\rho(u;\check\bbeta)-\rho(u;\tilde\bbeta)= (\check \bbeta - \tilde \bbeta )^\top \bz(u)\rho(u;\bbeta^\prime)$ whereby we deduce that  
			\begin{align*}
				\|\mathbf A_n(\check \bbeta)-\mathbf A_n(\tilde \bbeta)\| &\leq \int_{D_n} \|\bz(u)\|^3 \|\check \bbeta-\tilde \bbeta\| \rho(u;\bbeta^\prime) \dd u\\
				&=O_\P\left(\mu_n p_n^{3/2} \|\tilde \bbeta-\bbeta_0\|\right) = O_\P \left(p_n^2 \sqrt{\mu_n} \right).
			\end{align*}
			(iii) We only have to prove it for $\|\mathbf U_n(\tilde \bbeta)\|$. Using Taylor expansion, there exists $\check \bbeta$ such that $\mathbf U_n(\tilde \bbeta) = \mathbf U_n(\bbeta_0) - \mathbf A_n(\check \bbeta) (\tilde \bbeta-\bbeta_0)$. Using Lemma~\ref{bound}, (ii) and condition~\cond{C:initial}, we obtain
			\begin{align*}
				\|\mathbf U_n(\tilde \bbeta) \| = O_\P \left( \sqrt{p_n \mu_n} + p_n \mu_n \sqrt{p_n/\mu_n} \right) = O_\P\left( \sqrt{p_n^3\mu_n}\right).	
			\end{align*}
		\end{proof}

		\subsection{Sparsity property for $\hat \bbeta=\hat \bbeta_{\mathrm{ALDS}}$}

		The sparsity property of $\hat \bbeta_{\mathrm{ALDS}}$ follows from the following Lemma.
		
		\begin{lemma} \label{lemma:sparsity}
			Let $\hat \bbeta_{\mathrm{ALDS}}$ and $\hat \bgam$ satisfy the following conditions
			\begin{align}
				\hat \bbeta_{\mathrm{ALDS,1}} & =\mathbf{A}_{n,11}(\tilde {\bbeta})^{-1} \Big \{\mathbf{U}_{n,1}( \tilde{\bbeta}) + \mathbf{A}_{n,1}(\tilde {\bbeta})\tilde {\bbeta} - \mu_n\bld_{n,11}  \sign(\hat \bgam_1) \Big \} \label{eq:betahat1} \\
				\hat \bbeta_{\mathrm{ALDS,2}} & = \mathbf{0} \label{eq:betahat2} \\
				\hat \bgam_1 & = \mu_n \bld_{n,11}\mathbf{A}_{n,11}(\tilde {\bbeta})^{-1} \bld_{n,11} \sign(\hat \bbeta_{\mathrm{ALDS,1}}). \label{eq:gammahat1} \\
				\hat \bgam_2 & = \mathbf{0}. \label{eq:gammahat2}
			\end{align}
			Then, under the conditions~\cond{C:intensity}-\cond{C:anbn}, the following two statements hold.\\
			(i)
			\begin{align*}
				\hat \bbeta_{\mathrm{ALDS,1}} -\bbeta_{01} &= \mathbf A_{n,11}(\bbeta_0)^{-1} \mathbf U_{n,1}(\bbeta_0) + o_\P\left( \frac{1}{s_n\sqrt{\mu_n}} \right)  = 
				O_\P\left( \sqrt{\frac{s_n}{\mu_n}}\right).
			\end{align*}
			(ii) With probability tending to 1, $\hat \bbeta_{\mathrm{ALDS}}$ and $\hat \bgam$ given by \eqref{eq:betahat1}-\eqref{eq:gammahat2} satisfy conditions~\eqref{eq:fea1}-\eqref{eq:slack2} and are thus the primal and dual optimal solutions (whence the notation $\hat \bbeta_{\mathrm{ALDS}}$).
		\end{lemma}
		
		It is worth mentioning that the rate $o_\P(1/s_n\sqrt{\mu_n})$ in Lemma~\ref{lemma:sparsity} (i) is required to derive the central limit theorem proved in Appendix~\ref{sec:proofCLT_ALDS}. That required rate of convergence imposes some stronger restriction on the sequence~$a_n$.
		
		\begin{proof}
			(i) Using Taylor expansion, there exists $\check \bbeta$ on the line segment between $\tilde \bbeta$ and $\bbeta_0$ such that $\mathbf U_{n,1} (\tilde \bbeta) =\mathbf U_{n,1} (\bbeta_0) - \mathbf A_{n,1}(\check \bbeta)(\tilde \bbeta - \bbeta_0)$ which leads, by noticing that $\bbeta_{02}=0$, to
			\begin{align*}
				\hat \bbeta_{\mathrm{ALDS,1}}-\bbeta_{01} =&\mathbf{A}_{n,11}(\tilde {\bbeta})^{-1} 
				\bigg[ \left\{
				\mathbf{A}_{n,1}(\tilde {\bbeta}) -\mathbf{A}_{n,1}(\check {\bbeta})\right\} \left\{  \tilde {\bbeta}-\bbeta_0\right\} \\
				&+ \mathbf U_{n,1}(\bbeta_0) - \mu_n\bld_{n,11}  \sign(\hat \bgam_1) \bigg].
			\end{align*}
			Condition~\cond{C:cov} ensures that $\|\mathbf A_{n,11}(\tilde \bbeta)^{-1}\| =O_\P(\mu_n^{-1})$. Let $\hat \bbeta_{\mathrm{ALDS,1}}-\bbeta_{01} = \mathbf A_{n,11}(\bbeta_0)^{-1} \mathbf U_{n,1}(\bbeta_0) +  T_1+T_2+T_3$ where
			\begin{align*}
				T_1&=\left\{\mathbf A_{n,11}(\tilde\bbeta)^{-1} -\mathbf A_{n,11}(\bbeta_0)^{-1} \right\} \mathbf U_{n,1}(\bbeta_0) \\
				T_2&= \mathbf{A}_{n,11}(\tilde {\bbeta})^{-1} 
				\left\{
				\mathbf{A}_{n,1}(\tilde {\bbeta}) -\mathbf{A}_{n,1}(\check {\bbeta})\right\} \left\{  \tilde {\bbeta}-\bbeta_0\right\} \\
				T_3 &= \mu_n \mathbf{A}_{n,11}(\tilde {\bbeta})^{-1} \bld_{n,11}  \sign(\hat \bgam_1).
			\end{align*}
			Regarding the term $T_1$ we have
			\[
			T_1 = \mathbf A_{n,11}(\tilde\bbeta)^{-1} \left\{  \mathbf A_{n,11}(\bbeta_0)- \mathbf A_{n,11}(\tilde\bbeta)\right\} \mathbf A_{n,11}(\bbeta_0)^{-1} 
			\mathbf U_{n,1}(\bbeta_0).
			\]
			Condition~\cond{C:cov} ensures that $\max(\|\mathbf A_{n,11}(\tilde \bbeta)^{-1}\|,\|\mathbf A_{n,11}(\bbeta_0)^{-1}\|)= O_\P(\mu_n^{-1})$. Using this, Lemma~\ref{lem:aux} and Lemma~\ref{bound} we obtain
			\[
			T_1 = O_\P \left( \frac{1}{\mu_n} \sqrt{s_n^2 p_n^2 \mu_n} \frac{1}{\mu_n} \sqrt{s_n \mu_n}\right) = O_\P \left( \frac{\sqrt{s_n^3p_n^2}}{\mu_n}\right).
			\]
			With similar arguments, we have
			\[
			T_2 = O_\P \left( \frac{1}{\mu_n} \sqrt{s_np_n^3\mu_n}\sqrt{p_n/\mu_n}\right) = 
			O_\P \left( \frac{\sqrt{s_np_n^4}}{\mu_n}\right).
			\]
			Condition~\cond{C:snpn} ensures that
			\[
			T_1+T_2= O_\P\left( \frac{\sqrt{s_np_n^4}}{\mu_n}\right) = o_\P\left( \frac1{s_n \sqrt{\mu_n}}
			\right).
			\]
			Now, regarding the last term
			\[
			T_3 = O_\P\left( \mu_n \frac1{\mu_n} a_n \sqrt{s_n} \right) = O_\P (a_n \sqrt{s_n}). 
			\]
			And we observe that condition~\cond{C:anbn} is sufficient to establish that $T_3=o_\P(1/s_n\sqrt{\mu_n})$ which proves (i) using again condition~\cond{C:cov} and Lemma~\ref{bound}.
			
			(ii) We have to show that with probability tending to 1, $\hat \bbeta_{\mathrm{ALDS}}$ and $\hat \bgam$ given by \eqref{eq:betahat1}-\eqref{eq:gammahat2} satisfy conditions~\eqref{eq:fea1}-\eqref{eq:slack2}. 
			By \eqref{eq:betahat1}-\eqref{eq:gammahat2},
			\begin{align*}
				\mu_n^{-1}\hat \bgam^\top\bld_{n}^{-1} \mathbf{A}_n(\tilde {\bbeta}) \hat \bbeta_{\mathrm{ALDS}} &= \mu_n^{-1} \hat \bgam_1^\top\bld_{n,11}^{-1} \mathbf{A}_{n,11}(\tilde {\bbeta}) \hat \bbeta_{\mathrm{ALDS},1} \\
				& = \sign(\hat \bbeta_{\mathrm{ALDS,1}})^\top\bld_{n,11} \big\{\mathbf{A}_{n,11}(\tilde {\bbeta})\big\}^{-1} \bld_{n,11}\bld_{n,11}^{-1} \mathbf{A}_{n,11}(\tilde {\bbeta}) \hat \bbeta_{\mathrm{ALDS,1}} \\
				&= \sign(\hat \bbeta_{\mathrm{ALDS,1}})^\top\bld_{n,11}  \hat \bbeta_{\mathrm{ALDS,1}} \\
				& = \ \|\bld_{n,11} \hat \bbeta_{\mathrm{ALDS,1}}\|_1=  \|\bld_{n} \hat \bbeta_{\mathrm{ALDS}}\|_1,
			\end{align*}
			so, \eqref{eq:slack1} is satisfied. Now, we want to show that \eqref{eq:slack2} holds. We have
			\begin{align*}
				\mu_n^{-1} \hat \bgam^\top\bld_{n}^{-1} \big \{\mathbf{U}_{n}( \tilde {\bbeta}) +  \mathbf{A}_n( \tilde {\bbeta}) (\tilde {\bbeta}  -  \hat \bbeta_{\mathrm{ALDS}}) \big \} = \mathbf{I} + \mathbf{II},
			\end{align*}
			where
			\begin{align*}
				\mathbf{I} = \; & \mu_n^{-1} \hat \bgam^\top\bld_{n}^{-1} \mathbf{U}_{n}( \tilde {\bbeta}) = \mu_n^{-1} \hat \bgam_1^\top\bld_{n,11}^{-1} \mathbf{U}_{n,1}( \tilde {\bbeta}),\\
				\mathbf{II} = \;  & \mu_n^{-1} \hat \bgam_1^\top\bld_{n,11}^{-1}  \mathbf{A}_{n,1}( \tilde {\bbeta}) \tilde {\bbeta}  -  \mu_n^{-1} \hat \bgam_1^\top\bld_{n,11}^{-1}  \mathbf{A}_{n,11}( \tilde {\bbeta})  \hat \bbeta_{\mathrm{ALDS,1}} \\
				=  \; & \mu_n^{-1} \hat \bgam_1^\top\bld_{n,11}^{-1}  \mathbf{A}_{n,1}( \tilde {\bbeta}) \tilde {\bbeta} \\
				&-  \mu_n^{-1} \hat \bgam_1^\top\bld_{n,11}^{-1}  \{\mathbf{U}_{n,1}( \tilde {\bbeta}) + \mathbf{A}_{n,1}( \tilde {\bbeta}) \tilde {\bbeta} -\mu_n \bld_{n,11} \sign(\hat\bgam_1)\} \\
				=  \; & \hat \bgam_1^\top \sign(\hat\bgam_1)  -  \mu_n^{-1} \hat \bgam_1^\top\bld_{n,11}^{-1} \mathbf{U}_{n,1}( \tilde {\bbeta}), 
			\end{align*}
			from \eqref{eq:betahat1}-\eqref{eq:gammahat2}. By summing $\mathbf{I}$ and $\mathbf{II}$, we deduce that \eqref{eq:slack2} holds.
			
			To prove \eqref{eq:fea2} holds, we use \eqref{eq:gammahat2} and decompose the vector \linebreak$\mu_n^{-1} \bld_{n}^{-1} \mathbf{A}_n(\tilde {\bbeta})\bld_{n}^{-1} \hat \bgam $ as
			\begin{align*}
				\mu_n^{-1}  \bld_{n}^{-1} \mathbf{A}_n(\tilde {\bbeta})\bld_{n}^{-1} \hat \bgam = \mu_n^{-1}
				\begin{bmatrix}
					\mathbf{I}^\prime  \\
					\mathbf{II}^\prime
				\end{bmatrix}
				= \mu_n^{-1}
				\begin{bmatrix}
					\bld_{n,11}^{-1} \mathbf{A}_{n,11}(\tilde {\bbeta}) \bld_{n,11}^{-1} \hat{\boldsymbol \gamma_1 }  \\
					\bld_{n,22}^{-1} \mathbf{A}_{n,21}(\tilde {\bbeta}) \bld_{n,11}^{-1} \hat{\boldsymbol \gamma_1}
				\end{bmatrix}.
			\end{align*}  
			By~\eqref{eq:gammahat1}
			\begin{align*}
				\mu_n^{-1} \| \mathbf{I}^\prime \|_\infty & = \mu_n^{-1} \| \bld_{n,11}^{-1} \mathbf{A}_{n,11}(\tilde {\bbeta}){\bld_{n,11}^{-1}} \hat \bgam_1 \|_\infty \\
				& = \|  \sign(\hat \bbeta_{\mathrm{ALDS,1}}) \|_\infty =1.  
			\end{align*}
			Regarding $\mathbf{II}^\prime$, by \eqref{eq:gammahat1}, conditions on $a_n$ and $b_n$, conditions~($\mathcal{C}$.\ref{C:intensity})-($\mathcal{C}$.\ref{C:cov}), ($\mathcal{C}$.\ref{C:initial}) and Lemma~\ref{lem:aux}(i)-(ii), we have
			\begin{align*}
				\mu_n^{-1}\mathbf{II}^\prime & = \mu_n^{-1} \bld_{n,22}^{-1} \mathbf{A}_{n,21}(\tilde {\bbeta})\bld_{n,11}^{-1} \hat \bgam_1 \\
				& =\bld_{n,22}^{-1} \mathbf{A}_{n,21}(\tilde {\bbeta}){\bld_{n,11}^{-1}}\bld_{n,11}\mathbf{A}_{n,11}(\tilde {\bbeta})^{-1}\bld_{n,11} \sign(\hat \bbeta_{\mathrm{ALDS,1}}) \\
				& =\bld_{n,22}^{-1} \mathbf{A}_{n,21}(\tilde {\bbeta})\mathbf{A}_{n,11}(\tilde {\bbeta})^{-1}\bld_{n,11} \sign(\hat \bbeta_{\mathrm{ALDS,1}}) \\
				&=  
				O_\P \left( 
				\frac1{b_n} \sqrt{p_ns_n}\mu_n \frac1{\mu_n} a_n \sqrt{s_n}
				\right) = 
				O_\P \left(
				\frac{a_n \sqrt{s_n^2 p_n}}{b_n}
				\right) \\
				&=  O_\P \left( a_n \sqrt{s_n^3 \mu_n} \frac1{b_n}\sqrt{\frac{p_n^3}{\mu_n}} \, \frac{1}{s_np_n} \right). 
			\end{align*}
			Hence, $\mu_n^{-1}\|\mathbf{II}^\prime\|_\infty = o_\P(1)$ by condition~\cond{C:snpn} and
			\eqref{eq:fea2} is satisfied with probability tending to 1. We finally focus on \eqref{eq:fea1}. Note that
			\begin{align*}
				\mu_n^{-1} \bld_{n}^{-1}  \{\mathbf{U}_{n}( \tilde {\bbeta}) &+  \mathbf{A}_n( \tilde {\bbeta}) (\tilde {\bbeta}  -  \hat \bbeta_{\mathrm{ALDS}}) \}  = \mu_n^{-1}
				\begin{bmatrix}
					\tilde {\mathbf{I}}  \\
					\tilde {\mathbf{II}}
				\end{bmatrix} \\
				& = \mu_n^{-1}
				\begin{bmatrix}
					\bld_{n,11}^{-1}  \{\mathbf{U}_{n,1}( \tilde {\bbeta}) +  \mathbf{A}_{n,1}( \tilde {\bbeta}) (\tilde {\bbeta}  -  \hat \bbeta_{\mathrm{ALDS}}) \} \\
					\bld_{n,22}^{-1}  \{\mathbf{U}_{n,2}( \tilde {\bbeta}) +  \mathbf{A}_{n,2}( \tilde {\bbeta}) (\tilde {\bbeta}  -  \hat \bbeta_{\mathrm{ALDS}}) \} 
				\end{bmatrix}.
			\end{align*}
			Regarding $\tilde {\mathbf{I}}$, from~\eqref{eq:betahat1}-\eqref{eq:betahat2},
			\begin{align*}
				\mu_n^{-1} \|\tilde {\mathbf{I}}\|_\infty = & \mu_n^{-1} \|\bld_{n,11}^{-1}  \mathbf{U}_{n,1}( \tilde {\bbeta}) +\bld_{n,11}^{-1} \mathbf{A}_{n,1}( \tilde {\bbeta}) \tilde {\bbeta}  - \bld_{n,11}^{-1} \mathbf{A}_{n,11}( \tilde {\bbeta}) \hat \bbeta_{\mathrm{ALDS,1}} \|_\infty \\
				= & \mu_n^{-1} \|\bld_{n,11}^{-1}  \mathbf{U}_{n,1}( \tilde {\bbeta}) +\bld_{n,11}^{-1} \mathbf{A}_{n,1}( \tilde {\bbeta}) \tilde {\bbeta}  \\
				& - \bld_{n,11}^{-1} \{\mathbf{U}_{n,1}( \tilde {\bbeta})+ \mathbf{A}_{n,1}( \tilde {\bbeta}) \tilde {\bbeta} -\mu_n\bld_{n,11} \sign(\hat\bgam_1)\} \|_\infty \\
				= & \|\sign(\hat\bgam_1) \|_\infty = 1.  
			\end{align*}
			
			Now, consider $\tilde {\mathbf{II}}$. By the sparsity of $\hat \bbeta_{\mathrm{ALDS}}$ and $\bbeta_0$ we can write
			\[
			\mu_n^{-1} \tilde {\mathbf{II}} = \mu_n^{-1} \boldsymbol\Lambda_{n,22}^{-1} \left\{
			\mathbf U_{n,2}(\tilde \bbeta) + \mathbf A_{n,2}(\tilde \bbeta) (\tilde \bbeta-\bbeta_0) + 
			\mathbf A_{n,21}(\tilde \bbeta) (\bbeta_{01}-\hat \bbeta_{\mathrm{ALDS,1}})
			\right\}.
			\]
			We combine Lemma~\ref{lem:aux} (i)-(iii) and Lemma~\ref{sparsity} to derive
			\begin{align*}
				\mu_n^{-1}  \tilde {\mathbf{II}} &= O_\P 
				\left\{
				\frac{1}{\mu_n} \frac1{b_n} 
				\left(
				\sqrt{p_n\mu_n} + p_n \mu_n \sqrt{\frac{p_n}{\mu_n}} + \sqrt{p_n s_n} \mu_n \sqrt{\frac{s_n}{\mu_n}}
				\right)
				\right\}\\
				& =O_\P \left( \frac{1}{b_n} \sqrt{\frac{p_n^3}{\mu_n}}\right)  = o_\P(1)
			\end{align*}
			by condition~\cond{C:anbn}. Hence, $\mu_n^{-1}\| \tilde{\mathbf{II}} \|_\infty = o_\P(1)$ and
			\eqref{eq:fea1} is satisfied with probability tending to 1.
		\end{proof}

		\subsection{Asymptotic normality for $\hat \bbeta=\hat \bbeta_{\mathrm{ALDS}}$}
		\label{sec:proofCLT_ALDS}
		
		\begin{proof}
			By Lemma~\ref{lem:aux}, $\mathbf A_{n,11}(\bbeta_0)=O_\P(s_n\mu_n)$. This and Lemma~\ref{lemma:sparsity} show that
			\[
			\mathbf A_{n,11}(\bbeta_0) \left(\hat \bbeta_{\mathrm{ALDS,1}} -\bbeta_{01}\right)= \mathbf U_{n,1}(\bbeta_0) + o_\P(\sqrt{\mu_n}).  
			\]
			Let $\boldsymbol \phi \in \mathbb R^{s_n}\setminus\{0\}$ with $\|\boldsymbol \phi\|<\infty$ and let $\sigma^2_{\boldsymbol \phi} = \boldsymbol \phi^\top \mathbf B_{n,11}(\bbeta_0) \boldsymbol \phi$.  Now
			\[
			\sigma_{\boldsymbol \phi}^{-1}
			\boldsymbol \phi^\top
			\mathbf A_{n,11}(\bbeta_0) \left( \hat \bbeta_{\mathrm{ALDS,1}} -\bbeta_{01}\right) = 
			\sigma_{\boldsymbol \phi}^{-1}\boldsymbol \phi^\top
			\mathbf U_{n,1}(\bbeta_0) +
			\sigma_{\boldsymbol \phi}^{-1} \boldsymbol \phi^\top o_\P(\sqrt{\mu_n}).	
			\]
			By condition~\cond{C:Bn}, $\sigma_{\boldsymbol \phi}^{-1}=O(1/\sqrt{\mu_n})$.
			The result is therefore deduced  from condition~\cond{C:clt} and Slutsky's theorem.
		\end{proof}

		\section{Resulting $\bbeta$ estimates for the BCI dataset}
		
		\setlength{\tabcolsep}{1.35pt}
		\renewcommand{\arraystretch}{1.1}
		\begin{table}[H]
			\caption{Values of $\hat \bbeta_{\mathrm{AL}}$ and  $\hat \bbeta_{\mathrm{ALDS}}$ for the real data example}
			\label{tab:betaest}
			\centering
			\begin{tabular}{lrrrrrrrrrrrrrr}
				\hline
				& Int & elev & grad & Al & B & Ca & Cu & Fe & K & Mg & Mn & P & Zn & N \\ 
				\hline
				AL & -6.252 & 0 & 0 & 0 & -0.461 & 0.448 & 0 & 0.260 & 0 & 0 & -0.244 & 0 & 0 & 0 \\ 
				ALDS & -6.245 & 0 & 0 & 0 & -0.429 & 0.419 & 0 & 0.245 & 0 & 0 & -0.235 & 0 & 0 & 0 \\ 
				\hline
				& N.min & pH & AlB & AlCa & AlCu & AlFe & AlK & AlMg & AlMn & AlP & AlZn & AlN & AlN.min & AlpH \\ 
				\hline
				AL & 0.076 & 0.477 & -0.162 & 0 & 0 & 0.379 & 0 & -0.514 & 0 & -0.022 & 0 & 0.103 & 0 & -0.033 \\ 
				ALDS & 0.077 & 0.464 & -0.221 & 0 & 0 & 0.394 & 0 & -0.471 & 0 & -0.008 & 0 & 0.103 & 0 & 0 \\ 
				\hline
				& BCa & BCu & BFe & BK & BMg & BMn & BP & BZn & BN & BN.min & BpH & CaCu & CaFe & CaK \\ 
				\hline
				AL & 0 & 0 & 0.152 & 0 & 0 & 0 & -0.299 & 0 & -0.071 & 0 & 0 & 0 & 0.144 & 0 \\ 
				ALDS & 0 & -0.093 & 0.183 & 0 & 0 & 0 & -0.286 & -0.080 & -0.027 & 0.053 & 0 & 0 & 0.142 & 0 \\ 
				\hline
				& CaMg & CaMn & CaP & CaZn & CaN & CaN.min & CapH & CuFe & CuK & CuMg & CuMn & CuP & CuZn & CuN \\ 
				\hline
				AL & -0.155 & 0 & 0.104 & 0 & 0 & -0.888 & 0 & 0 & -0.091 & 0 & 0.134 & 0.148 & 0 & 0 \\ 
				ALDS & -0.125 & 0 & 0.095 & 0 & 0 & -0.890 & 0.042 & 0 & -0.013 & 0 & 0.130 & 0.148 & 0 & 0 \\ 
				\hline
				& CuN.min & CupH & FeK & FeMg & FeMn & FeP & FeZn & FeN & FeN.min & FepH & KMg & KMn & KP & KZn \\ 
				\hline
				AL & 0 & 0 & -0.311 & 0 & 0 & 0 & 0 & 0 & 0 & 0 & 0 & 0 & 0 & -0.051 \\ 
				ALDS & 0 & 0 & -0.331 & 0 & 0 & 0 & 0 & 0 & 0 & 0 & 0 & 0 & 0 & 0 \\ 
				\hline
				& KN & KN.min & KpH & MgMn & MgP & MgZn & MgN & MgN.min & MgpH & MnP & MnZn & MnN & MnN.min & MnpH \\ 
				\hline
				AL & 0.198 & 0.580 & 0.023 & 0 & 0 & -0.011 & 0 & 0 & 0 & 0 & 0 & -0.050 & 0.107 & 0 \\ 
				ALDS & 0.161 & 0.530 & 0 & -0.003 & 0 & 0 & 0 & 0 & 0 & 0 & 0 & -0.047 & 0.100 & 0 \\ 
				\hline
				& PZn & PN & PN.min & PpH & ZnN & ZnN.min & ZnpH & NN.min & NpH & N.minpH & elevgrad &&&\\ 
				\hline
				AL & 0 & 0.269 & 0 & 0 & 0 & 0 & 0 & 0 & 0 & 0.054 & 0 &&& \\ 
				ALDS & 0 & 0.258 & 0 & 0 & 0 & 0 & 0 & 0 & 0 & 0.054 & 0 &&& \\ 
				\hline
			\end{tabular}
		\end{table}
		
	\end{appendix}

	\section*{Acknowledgements}
	We thank the editor, associate editor, and two reviewers for the constructive comments. The research of J.-F. Coeurjolly is supported by the Natural Sciences and Engineering Research Council of Canada. J.-F. Coeurjolly would like to thank Université du Québec à Montréal for the excellent research conditions he received these last years. The research of A. Choiruddin is supported by the Direktorat Riset, Teknologi, dan Pengabdian Kepada Masyarakat, Direktorat Jenderal Pendidikan Tinggi, Riset, dan Teknologi, Kementerian Pendidikan, Kebudayaan, Riset, dan Teknologi Republik Indonesia. The BCI soils data sets were collected and analyzed by J. Dalling, R. John, K. Harms, R. Stallard and J. Yavitt with support from NSF DEB021104,021115, 0212284,0212818 and OISE 0314581, and STRI Soils Initiative and CTFS and assistance from P. Segre and J. Trani. 

	\bibliographystyle{imsart-number} 
	\bibliography{refthesis_2}

\end{document}